\documentclass{article}
\usepackage{graphicx} 
\usepackage{authblk}
\usepackage[T1]{fontenc}
\usepackage{soul, comment} 
\usepackage{phfqit}
\usepackage{proba} 
\usepackage[sort,nocompress]{cite}
\usepackage{amssymb}
\usepackage{amsmath}
\usepackage{amsthm}
\usepackage{amsfonts}
\usepackage{mathtools}
\usepackage{xspace}
\usepackage{bm}
\usepackage{nicefrac}
\usepackage{commath}
\usepackage{bbm}
\usepackage{boxedminipage}
\usepackage{xparse}
\usepackage[usenames,dvipsnames]{color}
\usepackage[dvipsnames]{xcolor}
\usepackage{float}
\usepackage{multirow}
\usepackage{makecell}
\usepackage{graphicx}
\usepackage{caption}
\usepackage{subcaption}
\usepackage{ifthen}
\usepackage{thm-restate}
\usepackage{algorithm}
\usepackage[noend]{algpseudocode}
\usepackage{colortbl} 
\usepackage{fancyhdr}   
\usepackage{hyperref}
    \hypersetup{ colorlinks=true, linkcolor=blue, citecolor=magenta, urlcolor=blue,}
\usepackage{fullpage}
\usepackage[utf8]{inputenc}
\usepackage{environ}
\usepackage{mdframed}
\usepackage{cleveref}
\usepackage[short]{optidef} 
\usepackage{enumitem}
\usepackage{tikz}
\usetikzlibrary{shapes}
\usetikzlibrary{positioning}
\usetikzlibrary{fit}
\usetikzlibrary{graphs,graphs.standard}
\usepackage{enumitem}

\newtheorem{lemma}{Lemma}[section]

\newtheorem{theorem}{Theorem}[section]
\newtheorem{corollary}{Corollary}[section]

\newtheorem{definition}{Definition}[section]
\newtheorem{conjecture}{Conjecture}[section]

\usepackage{enumitem}

\NewEnviron{problem}[1]{%
	\begin{center}\fbox{\parbox{6in}{%
				{\centering\scshape #1\par}%
				\parskip=1ex
				\everypar{\hangindent=1em}%
				\BODY
}}\end{center}}

\begin{document}

\bibliographystyle{alpha}

\title{On the Extension Theorem for Packing Steiner Forests}
\author{Jinghan A Zeng \footnote{Email: \texttt{jazeng2@illinois.edu}. Siebel School of Computing and Data Science, University of Illinois Urbana-Champaign. This work was done as my undergraduate thesis under the advisement of Chandra Chekuri.}}


\date{August 2026}

\maketitle

\begin{abstract}

We consider the problem of packing edge-disjoint Steiner forests in a graph.
The input consists of a multi-graph $G=(V,E)$ and a collection of $t$ vertex subsets $\mathcal{S} = \{S_1,S_2,\ldots,S_t\}$. A Steiner forest for $\mathcal{S}$, also called an $\mathcal{S}$-forest, is a forest of $G$ in which each $S_i$ is connected. In the case where $t=1$, this is the Steiner tree packing problem. Kriesell's conjecture postulates that $2k$-edge-connectivity of $S_1$ is sufficient to find $k$ edge-disjoint $S_1$-trees. Lau \cite{lau-thesis} showed that $24k$-edge-connectivity suffices for the Steiner tree packing problem, which was improved to $6.5k$ by West and Wu \cite{westwu} and $5k+4$ by Devos, McDonald and Pivotto \cite{devos2016packing}. In his thesis, Lau \cite{lau-thesis} asserts that for the Steiner Forest problem, if each $S_i$ is $30k$-edge-connected in $G$, then there exist $k$ edge-disjoint $\mathcal{S}$-forests. However, Lau's proof relies on an intermediate theorem called the Extension Theorem, which in this paper we will demonstrate has a gap by providing a counterexample to Lau's Extension Theorem. Furthermore, we will resolve this gap by correcting Lau's proof to show that $32k$-edge-connectivity of each $S_i$ suffices to pack $k$ $\mathcal{S}$-forests. More careful analysis yields that $31k$-edge-connectivity of each $S_i$ is sufficient when $k \geq 8$. 

\end{abstract}

\section{Introduction}
Let $G=(V,E)$ be a multigraph without self-loops. The well-known theorem of Tutte and Nash-Williams gives a min-max characterization for the maximum number of edge-disjoint spanning trees of $G$. This was later seen as a special case of packing disjoint bases in a matroid via the work of Edmonds. See \cite{Schrijver-book}. A simple and useful corollary of the Tutte-Nash-Williams theorem is that if $G$ is $2k$-edge-connected for some integer $k$, then $G$ has at least $k$ edge-disjoint spanning trees. Kriesell considered the problem of packing Steiner trees. Given $G$ and a set of \emph{terminal} vertices $S \subset V$ we call a subgraph $(V_H,E_H)$ of $G$ an 
$S$-tree if it connects $S$ (a minimal subgraph connecting $S$ would a tree). Kriesell conjectured that if $S$ is $2k$-edge-connected in $G$ then there exist $k$ edge-disjoint $S$-trees in $G$ \cite{Kriesell03}. This conjecture is still open, however much progress has been achieved starting with the important work of Lau \cite{lau2004approximate,lau-thesis}. Lau showed that it suffices for $S$ to be $24k$-edge-connected to guarantee $k$-edge-disjoint $S$-trees. West and Wu \cite{westwu} showed that $6.5k$-edge-connectivity suffices. Devos, McDonald and Pivotto further improved this bound to $(5k+4)$ \cite{devos2016packing}. 

We now consider the problem of packing Steiner forests which is a further generalization. Let $\mathcal{S} = \{S_1, S_2,\ldots,S_t\}$ be a collection of vertex-disjoint subsets. An $\mathcal{S}$-forest is a subgraph of $G$ in which each $S_i$ is connected (a minimal such subgraph will be a forest). Kriesell's conjecture admits a natural extension to this setting: if each $S_i$ is $2k$-edge-connected then $G$ has $k$-edge-disjoint $\mathcal{S}$-forests. The Steiner forest packing question was motivated initially by integer decomposition results \cite{chekuris09} who showed that the packing conjecture is true if $G$ is Eulerian. Lau \cite{lau2005packing,lau-thesis} asserted that if each $S_i$ is $30k$-edge-connected then $G$ admits $k$-edge-disjoint $\mathcal{S}$-forests. For this he relied on an extension theorem that builds on his work on packing Steiner trees. However, we will demonstrate that there is a gap in this extension theorem by providing a counterexample to the theorem. Furthermore, we will also remedy the gap in Lau's proof to obtain the following result:

\begin{theorem}
    Let $\mathcal{S} = \{S_1, S_2,\ldots,S_t\}$ be a collection of vertex-disjoint subsets in a multigraph $G$. If each $S_i$ is $32k$-edge-connected, then there exist $k$ edge disjoint $\mathcal{S}$-forests in $G$. 
\end{theorem}

We note that the main point of our paper will be to show that $Qk$ edge connectivity for some constant $Q$ suffices for the Steiner Forest packing problem, so we will show that $32k$-edge-connectivity suffices for simplicity. However, more careful analysis shows that $31k$-edge-connectivity suffices when $k \geq 8$. We also show some better bounds for when $t$ is small, that is, $9k$-edge-connectivity suffices when $t=2$ or $t=3$, and for $t > 3$, $(2t+3)k+Ct^2$-edge-connectivity suffices for some constant $C$. 

\section{Preliminaries}

We begin this section with some basic definitions. A set $S \subseteq V(G)$ on a graph is said to be $k$-edge-connected if after the deletion of any $k-1$ edges on $G$, we have that $S$ is still connected. By Menger's Theorem, $S$ being $k$-edge-connected is equivalent to being able to find $k$ edge-disjoint paths between any two vertices in $S$. An $S$-subgraph is a subgraph $H$ of $G$ such that $S$ is connected in $H$. An $(S_1,S_2...S_t)$-subgraph is a subgraph $H$ of $G$ such that for any $i \in [t]$, $S_i$ is connected in $H$. A minimal $S$-subgraph is referred to as a \emph{Steiner tree}, a minimal $(S_1,S_2...S_t)$-subgraph is referred to as a \emph{Steiner forest}. 

To find $k$ disjoint Steiner forests in a graph, both Lau's approach and our approach involve packing Steiner trees on some $S_i$, recursively packing Steiner Forests on the rest of the graph, and merging the two together to find a packing of $k$ Steiner forests. To formalize this idea, we will need to define notions of \emph{balancing} and an \emph{extension}, which we do below.  We let $E(v)$ denote the set of edges incident to $v$.

\begin{definition}
    Given a graph $G$ and a vertex $v$ with a subpartition $P_1 \sqcup P_2 ... \sqcup P_k \subseteq E(v)$ into $k$ parts where each part is allowed to be empty, we say that a collection of $k$ edge-disjoint connected subgraphs $H_1, H_2,...,H_k$ extend $v$ if all of the following criteria are true. 
    \begin{enumerate}
        \item For all $i$, $P_i \subseteq E(H_i)$. 
        \item For all $i$, $v$ is not a cut vertex of $H_i$. 
    \end{enumerate}
\end{definition}

We will define an operation called the \emph{edge-union} which uses the idea of an extension to "merge" two graphs together. We will prove a lemma that we will use frequently throughout the proof of the Extension Theorem, which provides motivation for how we define an extension. 

Take a graph $G$, and take some edge cut that splits the graph into components $C_1$ and $C_2$. We form two graphs $G_1$ and $G_2$, where $G_1$ is formed by taking $G$ and contracting $C_2$ into a single vertex $v_2$. Similarly, we form $G_2$ by taking $G$ and contracting $C_1$ into a single vertex $v_1$. Note that by construction, $E(v_1)=E(v_2)$, since both edge sets are the $(C_1,C_2)$-cut of $G$. 

Take some set $S \subseteq V(G)$, and denote $S_1 = S \cap G_1$ and $S_2 = S \cap G_2$. Suppose that we can find $k$ edge disjoint $S_2 \cup v_1$ subgraphs $H_1,H_2,...,H_k$ in $G_2$. Then, we take an edge subpartition of $v_2$ in $G_1$ to be $P_i = H_i \cap E(v_2)$. Suppose there exist $k$ edge disjoint $S_1 \cup v_2$-subgraphs $H_1', H_2',...,H_k'$ in $G_1$ that extend $v_2$. Then, we claim the following. For two graphs $H' \subseteq G_1$ and $H \subseteq G_2$, we use $H \uplus H'$ to denote the subgraph induced by $E(H) \cup E(H')$ in $G$. We refer to this operation as the \emph{edge-union} of two graphs. 

\begin{lemma}
    $H_1 \uplus H_1'$, $H_2 \uplus H_2'$... $H_k \uplus H_k'$ as defined above are $k$ edge disjoint $S$-subgraphs of $G$.
\end{lemma}

\begin{proof}
    We first prove that $H_1 \uplus H_1'$, $H_2 \uplus H_2'$... $H_k \uplus H_k'$ are edge disjoint. Suppose that there exists distinct $i,j$ such that $H_i \uplus H_i'$ and $H_j \uplus H_j'$ have a nonempty intersection. Then either $H_i \cap H_j' \not= \emptyset$ or $H_j \cap H_i' \not= \emptyset$, then without loss of generality, assume the former is true.
    Then, note that since the only edges $G_1$ and $G_2$ share are the edges in the $(C_1,C_2)$ cut, this implies that $H_i(v_1) \cap H_j'(v_2)$ is nonempty, where $H(v)$ denotes the edges of $H$ incident to $v$. However, by definition of extension and our construction, $H_i(v_1) \subseteq H_i'(v_2)$, hence $H_i'$ and $H_j'$ are not edge disjoint and we arrive at a contradiction. 

    We now claim that $H_i \uplus H_i'$ connects $S$ for any $i$. Note that $H_i'-v_2$ is connected, by the definition of an extension, and is a subgraph of $H_i \uplus H_i'$, hence for any $a,b \in S_1$, $a$ and $b$ are connected in $H_i \uplus H_i'$. Now suppose $a \in S_1$ and $b \in S_2$. Then, let $P_1$ be a path from $b$ to $v_1$ in $H_i$, let $e$ be the last edge in the path, and note that $e$ is incident to $v_1$. Let $v'$ be the endpoint of $e$ in $G_1$ that is not $v_2$. Then we can find a path $P_2$ from $v'$ to $a$ on $H_i'-v_2$. Then $P_1 \uplus P_2$ forms a path between $a$ and $b$ on $H_i \uplus H_i'$, hence $a$ and $b$ are connected. 
    
    Finally, suppose that $a,b \in S_2$. Take a path $P$ from $a$ to $b$ on $H_i$. If $P$ does not contain $v_1$, then $P$ exists on $H_i \uplus H_i'$ and we are done. Otherwise, if $v_1 \in P$, break $P$ into two parts $P_1$ and $P_3$, where $P_1$ is the path from $a$ to $v_1$ along $P$ and $P_3 = P-P_1$. Let $e$ be the last edge of $P_1$ and $e'$ be the first edge of $P_3$. Since both $e$ and $e'$ are incident to $v_1$, they appear on $G_1$ as well, and let $e=uv_2$ and $e' = u'v_2$ on $G_1$. Then, there exists a path $P_2$ from $u$ to $u'$ on $H_i'-v_2$, which we know is connected by the definition of an extension. Note that $P_1 \uplus P_2 \uplus P_3$ forms a path from $a$ to $b$ on $H_i \uplus H_i'$ and $a$ and $b$ are connected. Then this means that $S$ is connected on $H_i \uplus H_i'$, giving us our desired result. 
\end{proof}

The proof for Steiner forest uses a similar idea to this lemma, where we cut the graph into two components $C_1,C_2$ where $C_1$ is "small." We inductively pack $k$ Steiner forests on $C_2+v_1$, which will induce an edge subpartition of $v_2$, which we extend to $k$ edge disjoint Steiner trees on $C_1$ before using the edge-union operation to merge the packings together. In order to do this, we will need a theorem that allows us to pack Steiner trees which extend a (sub)partition on a vertex. This is the so-called \emph{Extension Theorem}. For this theorem to work, we need to assume that the partition on the vertex we extend is \emph{balanced}, which we define below. 

\begin{definition}
    Given a vertex $v$, a partition of $E(v) = P_1 \sqcup P_2 ... \sqcup P_k$ into $k$ parts is balanced if for all $i$, $|P_i| \geq 2$. 
\end{definition}

\begin{definition}
    Given a vertex $v$, a subpartition of $P_1 \sqcup P_2 ... \sqcup P_k \subseteq E(v)$ into $k$ parts, where a part is allowed to be the empty set, is called a balanced subpartition if there exists a balanced partition $E(v) = P_1' \sqcup P_2' ... \sqcup P_k'$ such that for all $i$, $P_i \subseteq P_i'$. 
\end{definition}

In other words, a subpartition of $v$ is balanced if it can be extended to a balanced partition of $v$. If an edge $e \in E(v)$ is part of $P_i$ for some (sub)partition, we say that the (sub)partition assigns $e$ a \emph{label} of $i$. 

\begin{definition}
    Given a graph $G$, a collection of $k$ edge disjoint subgraphs $H_1,H_2,...,H_k$ of $G$, we say that the $H_1,H_2,...,H_k$ balance $v$ if the edge subpartition $E(v) \cap H_1$, $E(v) \cap H_2$,..., $E(v) \cap H_k$ is a balanced edge subpartition of $v$. We say that $H_1,H_2,...,H_k$ balance a subset of vertices $S \subseteq V(G)$ if they balance all vertices in $S$.
\end{definition}

For any $e \in H_i$, we also say the packing assigns $e$ a \emph{label} of $i$. We now state the attempt at proving an Extension Theorem in Lau's thesis \cite{lau-thesis}. The following theorem was used to prove the bound of $30k$ for Steiner Forest Packing. Unfortunately, there is a mistake in this theorem which we elaborate on here. 

\begin{conjecture}[Lau's Extension Theorem]
    Let $G$ be a loopless multigraph, with $S,R \subseteq V(G)$ such that $S \cap R = \emptyset$, $S$ is $Qk$-edge-connected and for all $r \in R$, the degree of $r$ is at least $Qk$, where $Q \geq 30$. Given a vertex $v$ and a subpartition $P_k(v)$ of $E(v)$ into $k$ parts, there are $k$ edge-disjoint $S$-subgraphs that extend $v$ and balance $S \cup R$ if either the following is true.  
    \begin{enumerate}
        \item $v \in S$, $d(v) = Qk$, and $P_k(v)$ is a balanced edge subpartition of $E(v)$. 
        \item $N(v) \subseteq S \cup R$, $d(v) \leq Qk$, $v \not\in S \cup R$, and there is no edge cut of size at most $Qk$ that breaks $G$ into $C_1,C_2$, where $S \subseteq C_2$, $v \in C_1$, and $R \cap C_1 \not= \emptyset$. 
    \end{enumerate}
\end{conjecture}

The Extension Theorem is a very important tool in both packing Steiner Trees and packing Steiner Forests. Theorems \ref{thm:wwextthem} and \ref{thm:devosext} are examples of Extension Theorems used for packing Steiner Trees, and note that both of these theorems require some form of balancing at the vertex to be extended. One idea to tackle both Steiner Tree and Forest packing is to cut the graph into smaller parts, construct a packing on each part, and merge the packings together. We cut $G=(C_1,C_2)$ and contract into $C_1+v_2$ and $C_2+v_1$, construct a packing on $C_2+v_1$, extend $v_2$ to a packing in $C_1+v_2$, and edge-union the packings together.

As Extension Theorems usually require the vertex to be extended to start with a balanced partition, this means that we would like $v_2$ to start with a balanced partition, which is achieved when the packing on $C_2+v_1$ balances $v_1$. Compared to the Steiner Tree problem, however, achieving balancing on $v_1$ is much harder in the Steiner Forest problem. Hence, Lau's Extension Theorem proposes two conditions where either one guarantees an extension: 1) the vertex $v$ to be extended has a balanced partition, or 2) the neighborhood of $v$ is in $S \cup R$, along with some other conditions. The idea is that we can still apply the Extension Theorem on $v_2$ under the correct conditions even if we can't get $v_1$ to be balanced. 


Unfortunately, satisfying the second condition is not always sufficient for an extension of $k$ edge disjoint $S$-subgraphs to exist. We will disprove that the second condition is sufficient for a extension with a counterexample. We note that our counterexample is essentially the same as the one given by DeVos, McDonald, and Pivotto \cite{devos2016packing} in their paper to justify why balancing is needed for the Extension Theorem.

\begin{lemma}
    Let $k \geq 3$ be an integer. Given any integer $Q \geq 30$, there exists a loopless multigraph $G$ with $S \subseteq V(G)$, $Qk$-edge-connected, all vertices in $R \subseteq V(G)$ have degree at least $Qk$ and $S \cap R = \emptyset$, a vertex $v$ with a subpartition $P_k(v)$ of $E(v)$ that satisfies the second condition of Lau's Extension Theorem, but does not have $k$ edge disjoint $S$-subgraphs that extend $P_k(v)$. 
\end{lemma}

\begin{proof}
    Take two cliques $A$ and $B$, both of which are cliques on exactly $Qk+1$ vertices. Add $Qk-\lfloor (k-1)/2 \rfloor$ edges arbitrarily between $A$ and $B$, and set $S = A \cup B$. Select $Qk-k+1$ edges arbitrarily out of the edges between $A$ and $B$, and denote the set of these of edges $X$, and the set of the other edges between $A$ and $B$ as $Y$. Subdivide each edge in $X$. Then add a vertex $v$, and connect it to each of the new vertices we added to subdivide $X$. Add $\lfloor (k-1)/2 \rfloor$ edges to $v$ going to (not necessarily distinct) vertices in $A$ arbitrarily, and add $\lceil (k-1)/2 \rceil$ edges to $v$ going to (not necessarily distinct) vertices in $B$ arbitrarily. Add self-loops to the vertices that subdivide $X$ so that all of them have degree at least $Qk$, and subdivide the self loops themselves to ensure the graph is loopless. Designate all the vertices that subdivide $X$ to be in $R$. Note that $A \cup B$ is now $Qk$-edge-connected, as for any $a \in A$ and $b \in B$, $Qk-\lfloor (k-1)/2 \rfloor$ edge disjoint paths can be formed between $a$ and $b$ using $X \cup Y$, and $\lfloor (k-1)/2 \rfloor$ other edge disjoint paths can be formed from $a$ to $b$ by traversing through $v$. Furthermore, by construction, the degree of $v$ is exactly $Qk$. 

    Assume, for contradiction there is a cut of $G$ of size at most $Qk$ into $C_1,C_2$ such that $v \in C_1$, $S \subseteq C_2$, and $R \cap C_1 \not= \emptyset$. Let $R \cap C_1 = R_1$ and $R \cap C_2 = R_2$. Then, $v$ has $k-1+|R_2|$ edges incident to a vertex in $S$ or a vertex in $R_2$, which all must be part of the cut. Take any vertex in $r \in R_1$, and note that $r$ has two neighbors in $S$, hence both edges between $r$ and $S$ must also be part of the cut. This means that the cut has size at least $k-1+|R_2|+2|R_1| = Qk + |R_1| > Qk$, since $R_1$ must be nonempty, a contradiction. This means that this graph satisfies the second condition of Lau's Extension Theorem. 

    For all edges going from $v$ to $X$, designate them to have a label of $1$. For the edges going from $v$ to $A$, designate them to have labels $2,3,...,\lfloor (k-1)/2 \rfloor+1$. Finally, for the edges going from $v$ to $B$, designate them to have labels $\lfloor (k-1)/2 \rfloor+2,...,k$. Note that for any $H_i$ where $i \not=1$ to connect $A$ and $B$, at least one edge in $Y$ must be used, because it is impossible to use the subdivided edges in $X$, as they must be used in $H_1$ to connect some $r \in R$ to $A$ or $B$. However, note that $|Y| = Qk-\lfloor (k-1)/2 \rfloor - (Qk-k+1) = k - \lfloor (k-1)/2\rfloor-1 < k-1$ as long as $k \geq 3$. Hence, it is impossible to find $H_1,H_2,...,H_k$ that extend $v$, contradicting the theorem of Lau, as desired. 
\end{proof}

A figure illustrating the counterexample in the proof above is shown below.

    \begin{figure}[h]
        \centering
        \includegraphics[width=0.45\textwidth]{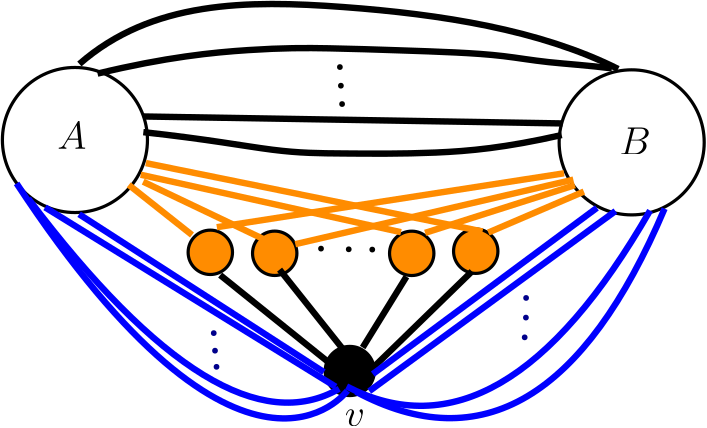}
    \end{figure} 

The vertices part of $R$ are highlighted in orange (the subdivided self loops are not pictured). The black edges incident to $v$ have a label of 1, while the blue edges incident to $v$ have labels $2$ to $k$ appearing exactly once. The orange edges are the subdivided edges of $X$ in the proof, and the black edges between $A$ and $B$ are the edges in $Y$. By construction, any subgraph extending a blue edge incident to $v$ must use one of the black edges between $A$ and $B$ to connect the subgraph, since the orange edges need to be used to extend the edges incident to $v$ with a label of 1. However, since by our construction, there are more blue edges than black edges between $A$ and $B$, such an extension is not possible. 
    
Lau's proof of the Extension Theorem is an induction proof that began by deleting vertices and edges that did not affect the connectivity of $S$. Unfortunately, one flaw with this approach is that it does not account for the fact that these deletions can make the degree of a vertex in $R$ fall below $Qk$, hence removing the vertex from $R$ and creating complications when recursing. To remedy this, we give and prove our own version of the Extension Theorem, which we will use to prove a bound of $32k$ for Steiner Forest Packing. 

In lieu of requiring the Extension Theorem to be true for a case similar to the second case in Lau's Extension Theorem, we drop the requirement for a vertex to be in $R$ from being incident to $Qk$ edges to $(Q-2)k$ edges. Even though self loops do not affect the connectivity of the graph, we still allow the graph to have them to make induction easier. We will prove the Extension Theorem for when $S$ is $32k$-edge-connected, but more careful analysis can yield a bound of $31k$ when $k \geq 8$.

\section{When $t$ is small}

Before we prove the Extension Theorem and our main theorem on Steiner Forest packing, we give a few results when $t$ is small. We will use the ideas in these proofs for our final result on Steiner Forest packing. When $t = 2$ or $t=3$, we can get a better bound of $9k$ when packing Steiner Forests. To prove this, we utilize the Extension Theorem of West and Wu. 

\begin{theorem}[West-Wu's Extension Theorem \cite{westwu}]\label{thm:wwextthem}
    Let $k$ and $\lambda k$ be positive integers with $\lambda \geq 6.5$. Suppose there is a loopless multigraph $G$ with some vertex $v$ such that $d(v) = \lambda k$, and $v \in S$ where $S \subseteq V(G)$ is $\lambda k$-edge-connected. Take an edge subpartition $E_1,E_2,...,E_k$ on $v$ into $k$ parts, and let $E_0$ denote the unlabeled edges incident to $v$. Then, if $|E_i| \geq 1$ for all $1 \leq i \leq k$ and $|E_0| \geq k$, there exist $k$ edge disjoint $S$-subgraphs $H_1,H_2,...,H_k$ on $G$ that extend $v$. Moreover, each vertex in $S$ is incident to at least $k$ edges which do not appear in $H_1 \cup H_2,...,H_k$, and $(H_1 \cup H_2... H_k) \cap E_0 = \emptyset$. 
\end{theorem}

The West-Wu Extension Theorem also implies the following main result of their paper.

\begin{theorem}[\cite{westwu}]\label{thm:wwpacking}
     Given a graph $G$ and $S \subseteq G$, if S is $6.5k$-edge-connected in $G$, then $G$ contains $k$ edge-disjoint $S$-subgraphs.
\end{theorem}

Finally, we will use the following lemma which will help us glue packings of Steiner trees together to pack Steiner forests. 

\begin{lemma}\label{lem:forestglue}
    Let $G$ be a graph with $S_1,S_2...,S_t \subseteq V(G)$, and suppose there is a cut $G=(C_1,C_2)$ such that $S_1 \subseteq C_1$ and $S_2,...,S_t \subseteq C_2$. Contract $C_1$ and $C_2$ to obtain $C_1+v_2$ and $C_2+v_1$. Suppose a packing of $k$ $(S_2,...,S_t)$-subgraphs exists on $C_2+v_1$, and copy the labels at $v_1$ over to $v_2$. If $C_1+v_2$ has $k$ edge-disjoint $S_1$-subgraphs that extend $v_2$, then $G$ has $k$ edge-disjoint $(S_1,S_2,...,S_t)$-subgraphs. 
\end{lemma}

\begin{proof}
    Let $H_1,H_2,...,H_k$ be the $k$ $(S_2,...,S_t)$-subgraphs on $C_2+v_1$ and let $H_1',H_2',...,H_k'$ be the $S_1$-subgraphs on $C_1+v_2$ that extend $v_2$. We claim that $H_1 \uplus H_1', H_2 \uplus H_2',..., H_k \uplus H_k'$ are $(S_1,S_2,...,S_t)$-subgraphs on $G$. We first note that for any $a,b \in S_1$ is connected in $H_i \uplus H_i'$ for any $i \in [k]$, as $H_i'-v_2$ is connected by the definition of an extension. Now take any $a,b \in S_j$ for $j \in \{2,...,t\}$, we wish to prove there is a path from $a$ to $b$ in $H_i \uplus H_i'$ for any $i \in [k]$. If a path from $a$ to $b$ in $H_i$ does not pass through $v_1$, then we are done. Otherwise, let $P_1$ be the segment of the path that starts at $a$ and first enters $v_1$, and let $P_3$ be the other segment of the path that leaves $v_1$ and reaches $b$. Let $e$ be the last edge traversed in $P_1$ and $e'$ be the first edge traversed in $P_3$, and find their corresponding edges $e=a'v_2$ and $e'=b'v_2$ on $C_1+v_2$. By construction, both edges must appear in $H_i'$, and we can find a path $P_2$ from $a'$ to $b'$ in $H_i'$. Then, we on $H_i \uplus H_i'$, we can start from $a$, traverse $P_1$, $P_2$, and then $P_3$ in that order, reaching $b$ and hence $a$ and $b$ are connected. As each $S_j$ is connected in each $H_i \uplus H_i'$, we have $k$ edge disjoint $(S_1,S_2,...,S_t)$-subgraphs on $G$, as desired.
\end{proof}

We now handle Steiner Forest packing for when $t=2,3$.

\begin{lemma}
    Given a graph $G$, let $S_1,S_2,S_3 \subseteq V(G)$ each be $9k$-edge-connected in $G$. Then there exists a packing of $k$ edge-disjoint $(S_1,S_2,S_3)$-forests. 
\end{lemma}

\begin{proof}
    We will first prove this statement when there are two groups, $S_1$ and $S_2$. Note that if $S_1 \cup S_2$ are $7k$-edge connected, then there exist $k$ edge disjoint $(S_1 \cup S_2)$-subgraphs by Theorem \ref{thm:wwpacking}, giving us our desired result. Otherwise, take the minimum edge cut $G=(C_1,C_2)$ of $S_1 \cup S_2$. Note that by minimality, less than $7k$ edges are in the cut, and hence $S_1$ and $S_2$ must both remain intact but be on different sides of the cut. Without loss of generality let $S_1 \subseteq C_1$ and $S_2 \subseteq C_2$. Contract $C_1$ into a vertex $v_1$ and $C_2$ into a vertex $v_2$. We can find $k$ edge-disjoint $S_2$-subgraphs on $C_2+v_1$ as $S_2$ is still $7k$-edge-connected on $C_2+v_1$. Let $E_1,E_2,...,E_k$ be a subpartition of $E(v_2)$ on $C_1+v_2$, where $E_i$ corresponds to the edges used by the $i$th subgraph on $C_2+v_1$. Note that $E_i$ may be empty. 

    Take any arbitrary $s \in S_1$, and add $9k-d(v_2)$ fake multi-edges from $v_2$ to $s$. We claim that $S_1 \cup v_2$ is now $9k$ connected. Indeed, assume there exists an edge cut of size $< 9k$. Note that $S_1$ will still be connected after such an edge cut, hence if any of the fake edges still exist $v_2$ will be connected to $S_1$. This means that all the fake edges must have been deleted, giving us $< d(v_2)$ edges to disconnect $v_2$ from $S_1$. However, note that if such a cut exists, this would imply the existence of a cut strictly less than $d(v_2)$ separating $C_1$ and $C_2$, contradicting the fact that we took the minimum cut. Since $d(v_2) < 7k$, we added at least $2k$ fake edges. For any $1 \leq i \leq k$, if $E_i$ is empty, take a fake edge and add it to $E_i$, and let $E_0 = E(v_2) - \bigcup_{i=1}^k E_i$, which denotes the unlabeled edges incident to $v_2$. Note that $|E_0| \geq k$, and the degree of $v_2$ is now exactly $9k$. Applying the extension theorem gives us $k$ edge disjoint subgraphs $H_1,H_2,...,H_k$ such that $S_1 \cup N_i$, where $N_i$ is the neighborhood given by $E_i$, is connected for each $H_i-\{v\}$. We can hence delete the fake edges and preserve the connectivity of $S_1 \cup N_i$, and hence $H_1,H_2,...,H_k$ still extend the original subpartition at $v_1$. Then by Lemma \ref{lem:forestglue}, we have the desired packing of $k$ $(S_1,S_2)$-subgraphs in $G$.

    Now assume there are three groups $S_1,S_2,S_3$, each of which are $9k$-connected. If $S_1 \cup S_2 \cup S_3$ are $7k$-connected, we are done. Otherwise, take a minimum cut $G=(C_1,C_2)$ that separates $S_1 \cup S_2 \cup S_3$, without loss of generality let $S_1 \subseteq C_1$ and $S_2 \cup S_3 \subseteq C_2$. Contract $C_1$ into $v_1$ and $C_2$ into $v_2$. Note that we can pack $k$ $(S_2,S_3)$-forests on $C_2+v_1$, and repeating a similar argument as above on $C_1+v_2$ allows us to extend that packing to $S_1$, giving us our desired result. 

\end{proof}

We can extend the ideas from the previous lemma to prove that $(2t+3)k + Ct^2$ connectivity suffices to pack $k$ Steiner trees when there are $t \geq 4$ groups for some universal constant $C$. This serves as a slightly better bound for the Steiner Forest packing problem when $t < 15$. 

\begin{lemma}
    There exists an increasing function $f(t)=O(t^2)$, such that for any integer $t \geq 3$ and $k \in \mathbb{N}$, given $S_1,S_2,...,S_t$ vertex-disjoint subsets of a graph $G$ where each $S_i$ is $(2t+3)k+f(t)$-edge-connected, there exists a packing of $k$ edge-disjoint $(S_1,S_2,...,S_t)$-subgraphs.
\end{lemma}

\begin{proof}
    We will prove this inductively; note that when $t=3$, $(2t+3)k=9k$, and hence a $f(3)=0$ makes the statement hold. Define $f(t) = f(t-1)+2t+1$ for $t \geq 4$. Let us take $S_1,S_2,...,S_t$, and suppose that given $t' < t$ groups in a graph all of which are $(2t'+3)k+f(t')$-edge-connected, we can pack $k$ $(S_1,S_2,...,S_t')$-subgraphs in the graph.

    If $S_1 \cup S_2...\cup S_t$ is $7k$-connected, we are done. Otherwise, take a cut $G=(C_1,C_2)$ of size $\leq (2t+3)k$ such that $|V(C_1)|$ is minimized and at least one terminal vertex exists in $C_1$. Let $t'$ be the number of vertex subsets from $S_1,S_2..., S_t$ in $C_2$, and without loss of generality let them be $S_1,S_2...,S_t'$. Contract $C_1$ and $C_2$ into $v_1$ and $v_2$ respectively. We wish to pack $k'=k+\lceil \frac{2}{2t+1}k \rceil$ edge disjoint $(S_1 \cup S_2...\cup S_t')$-subgraphs on $C_2+v_1$, and note that $k' \leq (1+\frac{2}{2t+1})k+1$. If $t' \geq 3$, to prove that we can pack $k'$ edge disjoint $(S_1 \cup S_2...\cup S_t')$-subgraphs, we need to show that $(2t'+3)k'+f(t') \leq (2t+3)k+f(t)$, since each $S_i$ is at least $(2t+3)k+f(t)$ edge connected on $C_2+v_1$.  We know that $t' \leq t-1$, hence $(2t'+3)k'+f(t') \leq (2(t-1)+3)((1+\frac{2}{2t+1})k+1)+f(t-1) = (2t+3)k+f(t)$, meaning that such a packing is possible. Otherwise, if $t' \leq 2$, we just need that each $S_i$ is $9k'$-edge-connected, but by our previous calculation $(2t+3)k+f(t)$ is at least $9k' = (2(3)+3)k'$ for any $t \geq 4$, as desired. 

    We now prove that all terminal vertices in $C_1+v_2$, that is, all vertices in $(S_1 \cup S_2... \cup S_t) \cap C_1$, are $(2t+3)k$-edge-connected. Indeed, suppose for contradiction there is an edge cut of size smaller than $(2t+3)k$ separating $C_1+v_2=(D_1,D_2)$, where $D_1$ and $D_2$ both contain terminal vertices, without loss of generality we assume $v_2 \in D_2$. Then this edge cut separates $D_1$ from $G-D_1$ as well, contradicting the fact that $C_1$ was minimized. Hence, we can treat all terminal vertices as a single group, which we will call $S'$. Take a vertex $s \in S'$, and add fake multiedges until the degree of $v_2$ is exactly $(2t+3)k$. We claim that $S' \cup v_2$ is $(2t+3)k$-edge-connected. 
    Suppose we were able to find such a cut $(D_1,D_2)$ of size smaller than $(2t+3)k$ where $S' \subseteq D_1$ and $v_2 \in D_2$. Note that $D_2$ cannot consist of only $v_2$, as there were exactly $(2t+3)k$ edges (including fake ones) incident to $v_2$, and all of them must have been deleted for this to happen. Hence, this means $D_1$ is strictly smaller than $C_1$, giving us an edge cut of at most $(2t+3)k$ smaller than $C_1$, a contradiction. 
    
    Let $E_0,E_1,...,E_{k+\lceil 2k/(2t+1) \rceil}$ partition $E(v_2)$, where for $1 \leq i \leq k+\lceil 2k/(2t+1) \rceil$, $E_i$ corresponds to the edges used for $i$th subgraph on $C_2+v_1$, and $E_0 = E(v_2) - \bigcup_{i=1}^{k+\lceil 2k/(2t+1) \rceil} E_i$. Re-index $E_1,...,E_{k+\lceil 2k/(2t+1) \rceil}$ so that $|E_i| \leq |E_j|$ if $i \leq j$, swapping the names of the respective subgraphs in $C_2+v_1$ as well. We now completely erase any labels $i > k$ from the graph, only paying attention to the first $k$ Steiner forests we packed on $C_2+v_1$. Note that consequently, $E_0$ gets updated with $E'_0 \leftarrow E_0 + \bigcup_{i=k+1}^{k+\lceil 2k/(2t+1) \rceil} E_i$, and $k+\frac{2k}{2t+1} = \frac{2tk+k}{2t+1}+\frac{2k}{2t+1} = \frac{(2t+3)k}{2t+1}$, so a proportion of at least $\frac{2}{2t+3}$ of the edges at $v_2$ are now a part of $E_0$. We then note that $|E_0| \geq (\frac{2}{2t+3})(3+2t)k = 2k$. For any empty $E_i$, take an edge from $E_0$ and move it to $E_i$. Now, note that the degree of $v_2$ is exactly $(3+2t)k$, $|E_0| \geq k$, and $E_1,...,E_k$ are nonempty, hence we can apply the West-Wu Extension Theorem to get $k$ edge-disjoint edge $S'$-subgraphs that extend $v_2$. By the properties of an extension, we can delete all the fake edges and still maintain that these subgraphs connect $S'$ and extend the partition at $v_2$ restricted to the non-fake edges. Then, by Lemma \ref{lem:forestglue}, we gain $k$ edge-disjoint $(S_1,S_2,...,S_t)$-subgraphs by the edge-union operation, as desired. 
\end{proof}

\section{Statement and Proof of Extension Theorem}

\begin{theorem}[Extension Theorem]
    Suppose there is a multigraph $G$ that may admit self loops, and $S, R \subseteq V(G)$ with $S \cap R = \emptyset$, $|S| \geq 2$, $Q \geq 32$ where $Q$ is an integer, and there exists a vertex $v \in S$ incident to no self loops with $d(v)=Qk$ and a balanced edge subpartition $P_k(v)$. If $S$ is $Qk$-edge-connected in $G$ and every vertex in $R$ is incident to at least $(Q-2)k$ edges, then there are $k$ edge-disjoint $S$-subgraphs that balance $S \cup R$ and extend $P_k(v)$.
\end{theorem}
Suppose, for contradiction, that this theorem is false. We will pick a counterexample $G$ with this criteria:

\begin{enumerate}
    \item Pick the counterexample where the sum of the vertices and edges that are not self loops is minimized. 
    \item If there are multiple such graphs, tiebreak by picking the counterexample with $|R|$ minimized. 
    \item If there are still multiple such graphs, tiebreak by picking the counterexample with $|S|$ maximized. If tiebreaking is still needed then choose arbitrarily. 
\end{enumerate}

We may assume without loss of generality, that each label $1,2,...,k$ appears at least twice on $E(v)$, because we can assign labels to the otherwise unlabeled edges if need be. Since we will be preserving the degree of vertices in $R$ by adding self loops after the deletion of edges, each self loop only contributes 1 instead of 2 to the degree of the vertex.

Our proof of the Extension Theorem will be an inductive or recursive proof formulated as a proof by contradiction. It will be very similar to Lau's original proof of the Extension Theorem, but as both this and Lau's proof are inductive in nature, due to the differences in how $R$ is defined, very few parts of Lau's original proof can actually be used as a black box. However, many parts of our proof will still be very similar to Lau's original proof, and so we will indicate before the proof of those lemmas what part of Lau's proof it resembles. For the sake of this proof being self-contained, we will still fully prove those lemmas, including the few parts of Lau's original proof that can be black-boxed. 

As the proof of the Extension Theorem is rather long, we will describe an outline of the proof before jumping in. The proof of the Extension Theorem is split into three parts. Part 1 starts by recursing down on $G$, simplifying $G$ until all vertices in $V(G)-S$ form an independent set, each vertex in $V(G)-S-R$ has degree 3 and exactly 3 neighbors, and $S \cup R$ is $(Q-4)k$-edge-connected. 

The next two parts concern constructing the $S-$subgraphs that give an extension, which will take advantage of the established structure in Part 1. We will handle the case where $|S|=2$ by using a collection of edge-disjoint paths between the two vertices in $S$. Otherwise, we will attempt to delete $v$ and its neighbors not in $S \cup R$ to obtain a new graph $G'$, pack $(S \cup R-v)$-subgraphs on $G'$, and add $v$ and its neighbors back to complete the extension. If we let $Z$ be the minimum cut of $S \cup R-v$ in $G'$, we will find that $|Z| \geq 6k$ is sufficient for this strategy to succeed. These two cases are covered in Part 2 of the proof.  

Part 3 concerns the case where $|Z| < 6k$, which means that the edge-connectivity of $S \cup R-v$ is too low after the deletion of $v$ and its neighbors for the strategy in Part 2 to work. We will then delete a small number of edges to disconnect $S \cup R-v$ in $G'$ and prove that $S \cup R-v$ is now highly edge-connected in each component. We then pack $S \cup R-v$-subgraphs in each component and use various tools to merge the packing together, which will complete the proof for the Extension Theorem. 

\subsection{Extension Theorem Proof Part 1: Recursing down on $G$}

We begin our recursion, where we will start by showing we can eliminate cut vertices outside of $S \cup R$ from the graph. 

\begin{lemma}
    No vertex in $V(G)-S-R$ is a cut vertex. 
\end{lemma}

\begin{proof}
    Suppose, for contradiction, that there is such a cut vertex, and denote that vertex $w$. Suppose that the deletion of $w$ splits $G$ into components $C_1, C_2,...,C_s$. Add a copy of $w$ to each component, denoted $w_1,w_2,...,w_s$, where the edges between $w_i$ and $C_i$ are the edges between $w$ and $C_i$ in $G$. If exactly one of these components contains vertices in $S$, let that component be $C_1$ without loss of generality. We then inductively apply the extension theorem on $w_1+C_1$ - which we may do since $w_1+C_1$ is strictly smaller than $G$. Note that this automatically gives us $k$ edge disjoint $S$ subgraphs on $G$ that balance $S \cup R$ and extend $v$, and note the vertices of $R$ which are not in $C_1$ are balanced since none of their incident edges were used in the packing. 

    Now, we handle the case where more than one of these components contains vertices in $S$. We claim that $S \cup w$ is $Qk$-edge-connected. Suppose not. Then there exists a vertex $s \in S$ which is separated from $w$ with a cut of less than $Qk$, and without loss of generality let $s \in C_1$. But there also exists some vertex $s' \in (S \cap (V(G)-C_1))$, and hence $s$ needs to traverse through $w$ to reach $s'$. If we can separate $s$ and $w$ with an edge cut of less than $Qk$, this will also separate $s$ and $s'$, meaning that $S$ is not $Qk$-edge-connected, a contradiction. Create $S'$ by adding $w$ to $S$, and note that $S'$ is $Qk$-edge-connected, $v \in S'$ still has a balanced edge subpartition, and every vertex in $R$ is still incident to at least $(Q-2)k$ edges. Since $G$ did not change and the size of $R$ did not change, by the third criteria of picking our counterexample, the extension theorem will hold if we increase the size of $S$. Hence there must exist $k$ edge-disjoint $S'$-subgraphs that balance $S' \cup R$ and extend $v$. As these subgraphs connect $S$ and balance $S \cup R$, we have arrived at a contradiction, and hence there can be no cut vertex in $V(G)-S-R$.

\end{proof}

We will now show that all vertices in $V(G)-S-R$ have degree 3 and are adjacent to 3 neighbors. To do this, we use Mader's Splitting Off Lemma, which allows us to replace two edges incident to a vertex with a single edge under the right circumstances, reducing the number of edges in the graph.

\begin{theorem}[Mader's Splitting Off Lemma \cite{Mader}]
    Let $x$ be a vertex of a loopless multigraph $G$. Suppose that $x$ is not a cut vertex and that $x$ is incident with at least 4 edges and adjacent to at least 2 vertices. Then there exists edges $xy$ and $xz$ where $y \not= z$, such that the number of edge disjoint paths between $a,b \in V(G)-x$ in $G'-xy-xz+yz$ is at least the number of edge disjoint paths between $a$ and $b$ in $G$.
\end{theorem}

\begin{lemma}
    All $ u\in V(G)-S-R$ have degree 3 and are adjacent to 3 vertices. 
\end{lemma}

\begin{proof}

    Take a vertex $u \in V(G)-S-R$. We can assume, without loss of generality, that $u$ is not incident to any self loops, because deleting the self loops do not affect the connectivity of $S$, nor the number of edges each vertex in $R$ is incident to. 
    
    If $u$ is adjacent to only one vertex $w$, delete it, and note that the deletion does not affect the connectivity of $S$ in $G-u$. If $w$ is a vertex in $R$, for each edge between $u$ and $w$, we add a self loop to $w$ corresponding to the deleted edge so that $w$ is still incident to at least $Qk$ edges. We can then inductively find $k$ edge disjoint $S$-subgraphs that extend $v$ balance $S \cup R$ on $G-u$. If $w$ was in $R$, we replace each of the added self loops with the corresponding original edge on $G$, which does not affect the connectivity of any of those $S-$subgraphs nor the balancing. Note that it is not possible for $w = v$, because if that were true, deleting any edge incident to $v$ would make $S$ no longer $Qk$-connected, but deleting $u$ and its incident edges would not affect the connectivity of $S$.  

    If $u$ is adjacent with exactly two edges, $uy$ and $uz$ with $y \not= z$, then delete both edges. Add another edge between $y$ and $z$. We note that we did not affect the connectivity of $S$ nor the degree of any vertex in $R$, so we can inductively find $k$ disjoint $S-$subgraphs that extend $v$ and balance $S \cup R$ on $G-uy-uz+yz$. If the added edge $yz$ appears in one of the $S-$subgraphs, then delete it and replace it $uy$ and $uz$ to get a packing on $G$, which preserves both connectivity and balancing. Similarly, if $u$ is incident to at least 4 edges, since $u$ is incident to no self loops, then we can delete the self loops from this graph and apply Mader's Splitting Off Lemma at $u$, finding $uy$ and $uz$ with $y \not= z$ such that $S$ is $Qk$-connected in $G-uy-uz+yz$. We then add the deleted self loops back and follow a similar argument as before to obtain the desired result.  

    If $u$ is incident with exactly 3 edges but only has two neighbors, $y$ and $z$, then note that any path that goes through $u$ has to pass through $y$ and $z$ and use only two edges. Without loss of generality, say that $u$ and $y$ share two edges. Hence, deleting one edge between $u$ and $y$ will not affect the connectivity of $S$. We can hence delete one edge between $u$ and $y$, and if $y \in R$ add one self loop to $R$ so it is still incident to at least $(Q-2)k$ edges. Note that $y$ cannot be the same as $v$, because it would contradict $v$ having degree $Qk$ yet being in $S$. Hence we can find $k$ edge disjoint $S$-subgraphs that extend $v$ and balance $S \cup R$ on $G-uy+yy$, and we replace the added self loop with the corresponding edge between $u$ and $y$ on $G$, giving us our desired result. 
\end{proof}

We will now show that $V(G)-S$ must be an independent set, which is done so by applying the Extension Theorem recursively and using the edge union operation. We note that one of the issues with the edge union operation has is that it does not preserve balancing. Suppose we split $G = (C_1,C_2)$ and contract each component into $C_1+v_2$ and $C_2+v_1$. If we have a packing on $C_1+v_2$ and copy the labels over to $v_1$, when constructing an extension on $C_2+v_1$, an unlabeled edge $e$ at $v_1$ might gain a label, and hence retain that label during the edge union operation. The vertex in $C_1+v_2$ incident to $e$ could be balanced in the packing at $C_1+v_2$ but not in $G$, because it gained a label at $e$ which could destroy its balancing. Hence, we define a \emph{strict edge extension}.

\begin{definition}
    Given a graph $G$ and a vertex $v$ with a subpartition $P_1 \sqcup P_2 ... \sqcup P_k \subseteq E(v)$ into $k$ parts where each part is allowed to be empty, we say that a collection of $k$ edge-disjoint subgraphs $H_1, H_2,...,H_k$ is a strict extension of $v$ if $H_1,H_2,...H_k$ extend $v$ and for all $i \in [k]$, any edge incident to $v$ and in $H_i$ is in $P_i$. 
\end{definition}

Essentially, a strict extension does not allow the unlabeled edges incident to $v$ to take on a new label when $v$ is extended. This is useful for the edge union operation, as it ensures each vertex ends with the same set of labels as incident edges that it started out with, allowing balancing to be preserved. The following lemma outlines how we will change our extensions to strict extensions. 

\begin{lemma}\label{lem:strict-extn}
    Take a graph $G$, a subset $S \subseteq V(G)$ with $|S| \geq 2$, a vertex $v \in S$, and a subpartition $P_1 \sqcup P_2 ... \sqcup P_k \subseteq E(v)$ where each part \emph{is nonempty}. Suppose $G$ has $k$ edge-disjoint $S-$subgraphs that extend $v$, and let $B$ be the balanced vertices in $G$ from this extension. Then, there exists $k$ edge-disjoint $S$-subgraphs that balance $B$ and form a strict edge-extension of $G$. 
\end{lemma}

\begin{proof}
    Take $k$ edge-disjoint $S-$subgraphs $H_1,H_2,...,H_k$ that extend $v$ and balance $B$. Then, for each $H_i$, if any edge $e \in H_i$ is incident to $v$ but not in $P_i$, remove $e$ from $H_i$, and denote the new subgraph as $H_i'$. We claim that $H_1', H_2',...,H_k'$ are $S$-subgraphs that balance $B$ and extend $v$. We first note that $H_1',H_2',...,H_k'$ must balance $B$, since removal of labels of incident edges at a vertex cannot change a vertex from balanced to non-balanced. We then note that for each $H_i'$, we have that $H_i'-v$ is connected, because by the definition of an extension, $v$ is not a cut vertex in $H_i$. Since we stipulated that $P_i$ is nonempty, there exists some $vu \in P_i$, and since $u \in H_i'-v$, $H_i'$ connects $v$ to $S-v$ as $S-v \subseteq H_i'-v$, and hence $v$ can traverse through $u$ to reach the other vertices in $S$. Then, by construction, $H_1',H_2',...,H_k'$ are $S$-subgraphs that strictly extend $v$ and balance $B$.
\end{proof}

\begin{lemma}\label{lem:indset}
    $V(G)-S$ is an independent set. 
\end{lemma}

\begin{proof}
    Suppose not, and there is an edge $e$ between the vertices in $V(G)-S$. If deleting that edge $e$ keeps $S$ $Qk$-edge-connected, then we delete $e$ from the graph. If $e$ was incident to any vertices in $R$, we add a self loop to those incident vertices to preserve their degree. We can then recursively find $k$ $S$-subgraphs $H_1,H_2...H_k$ that balance $S \cup R$ and extend $v$ on the new graph $G'$. Note that we can assume no self loops were used in the packing, because removing a self loop from $H_i$ does not affect the connectivity of $H_i$. Moreover, if a vertex $u$ is balanced and incident to a self loop in $H_i$, then removing the incident self loop from $H_i$ will still keep $u$ balanced in $G'$. Hence, we can assume that any self loops we artificially added after we deleted $e$ have no label. We can then delete those self-loops we added and add $e$ back to $G$, ensuring that $e$ has no label as well. Then any vertex incident to $e$ will maintain the exact same set of labels on its incident edges, ensuring that if the vertex was in $R$, it is still balanced on $G$. Hence, this gives us $k$ edge-disjoint $S$-subgraphs that balance $S \cup R$ and extend $v$, as desired. 
    
    If $S$ is no longer $Qk$-edge connected after the deletion of $e$, then there must be a collection of exactly $Qk$ edges (including $e$) which cuts $S$. Suppose that deleting those $Qk$ edges splits $G$ into $C_1$ and $C_2$. Contract $C_2$ into a single vertex $v_2$ and $C_1$ into $v_1$. We can assume without loss of generality that $v \in C_1$. Note that the degrees of all vertices are preserved on $C_1+v_2$ and $C_2+v_1$. Furthermore, $S \cap C_1$ is $Qk$-edge-connected in $C_1+v_2$ and $S \cap C_2$ is $Qk$-edge-connected in $C_2+v_1$. 

    Since $C_2$ had to contain one of the endpoints of $e$ and contains a vertex in $S$, we know that $C_1+v_2$ is strictly smaller than $G$. We also note that $v_2$ must be $Qk$-connected to $C_1 \cap S$, as there are $Qk$ edge disjoint paths from any vertex in $C_1 \cap S$ to any vertex in $C_2 \cap S$, which all must pass through the edge cut of $Qk$ edges. Hence, we can apply the extension theorem to find $k$ edge disjoint $S \cup v_2$-subgraphs $H_1, H_2,...,H_k$ on $C_1+v_2$ which extend $v$ and balance $(C_1 \cap (S  \cup R)) \cup v_2$. This induces a balanced edge subpartition on $v_1$, and note that each part is nonempty, since $v_2$ must have an incident edge in each $H_i$ by definition of an $S \cup v_2$-subgraph, so all $k$ labels must appear at $v_2$ and are subsequently copied over to $v_1$.
    
    We apply the extension theorem on $C_2+v_1$ in a similar manner to get $k$ edge disjoint $S \cup v_1$-subgraphs $H'_1, H'_2,...,H'_k$ which extend $v_1$ and balance $(C_2 \cap (S \cup R)) \cup v_1$ on $C_2+v_1$. By Lemma \ref{lem:strict-extn}, we can assume that this extension is also a strict extension. Then, taking edge-unions $H_1 \uplus H_1', H_2 \uplus H_2'...H_k \uplus H_k'$ gives us $k$ edge disjoint $S$ subgraphs on $G$ which extend $v$. We note that $S \cup R$ is balanced, because as we used a strict extension on $C_2$, no vertex had an incident edge change labels during the edge union, and hence all balanced vertices in $H_i$ and $H_i'$ are still balanced in $H_i \uplus H_i'$, giving us our desired result. 
\end{proof}


The subsequent two lemmas will demonstrate how to recurse to a point where $S \cup R$ is highly edge-connected. From there we will establish enough structure on $G$ to construct the $S$-subgraphs. 

\begin{lemma}
    Every vertex in $R$ has less than $2k$ self-loops.
\end{lemma}

\begin{proof}
    Suppose not, and let $R'$ be the set of all vertices in $R$ that are incident to less than $2k$ self loops. Note that $R'$ is a strict subset of $R$, so by minimality of $R$ (the second criteria of picking our counterexample), there exist $k$ edge disjoint $S$-subgraphs that extend $P_k(v)$ and balance $S \cup R'$. Take a vertex $r \in R-R'$. If any self loops incident to $r$ were used in any of the $S$-subgraphs, we can remove them from the $S-$subgraph, and note that their removal does not affect the connectivity of $S$. Since at least $2k$ self loops are incident to $r$ and none of them are used in any $S$-subgraph by construction, then $r$ is balanced as well. Hence, this is a contradiction, and all vertices in $R$ must have less than $2k$ self-loops.
\end{proof}

The following lemma then follows as a corollary of the previous two lemmas. 

\begin{lemma}\label{lem:Q-4-conn}
    $S \cup R$ is $(Q-4)k$-edge-connected in $G$. 
\end{lemma}

\begin{proof}
    Suppose not, and take an edge cut of $S \cup R$ of size less than $(Q-4)k$. We note that since $S$ is $Qk$-edge-connected, all of $S$ is contained in one component after the edge cut. Take some $r \in R$ in a different component after the cut. Since $r$ was incident to at most $2k$ self loops, then it must have at least $(Q-4)k$ incident edges that are incident to other vertices. However, since $r$ can only be incident to vertices in $S$ by virtue of $V(G)-S$ being an independent set, it means that it had $(Q-4)k$ edges incident to a vertex in $S$ which must have all been deleted in the cut. But we assumed the cut was less than $(Q-4)k$, a contradiction. Hence $S \cup R$ is $(Q-4)k$-edge-connected. 
\end{proof}

\subsection{Extension Theorem Proof Part 2: Constructing $S$-subgraphs with $|Z| \geq 6k$}

Now that we have established the structure of $G$, we will start explicitly constructing the $S-$subgraphs. This section will handle the case where $|S|=2$ and $S \cup R-v$ is at least $6k$-edge-connected after $v$ and its neighbors outside $S \cup R$ is deleted (in other words, the minimum cut $Z$ of $S \cup R-v$ in the modified graph has size at least $6k$). The rest is handled in Part 3 of the proof. We start by resolving our first case $|S|=2$.

\begin{lemma}
    $|S| \geq 3$.
\end{lemma}
\begin{proof}
    This proof is essentially the same as the section "When $|S| \leq 2$" in \cite{lau-thesis}. We consider the case where $|S|=2$ (and note that we assume $|S| \geq 2$, so eliminating this case leaves us with $|S| \geq 3$). In this case, $S$ consists of $v$ and another vertex $u$. Take any collection of $Qk$ edge disjoint paths between $v$ and $u$. Denote $P_i = \{P_{i,1}, P_{i,2},...,P_{i,m_i} \}$ as the paths from $v$ to $u$ that start by traversing an edge with label $i$ from $v$. Note that $P_i$ has a size of at least 2 since we assumed each label appears twice at $v$. 

    Construct $H_i$ by taking $P_{i,1} \cup P_{i,2}$. Then, for any remaining paths in $P_i$, iteratively add them to $H_i$ by starting at $v$, traversing the path to $u$, and truncating immediately when they reach a vertex of $H_i$ next, essentially adding an ear to $H_i$. Note that by this construction, $u$ and $v$ are connected, $H_i$ remains connected even after $v$ is deleted, and the degree of $u$ and $v$ in $H_i$ is at least 2, balancing both vertices. 

    We now prove that $R$ is balanced by claiming for each vertex $r \in R$, at least $2k$ edges go unused in $H_1, H_2,...,H_k$. Take a vertex $r \in R$, and assume it appears in $l$ paths/ears in the construction of $H_1, H_2,...,H_k$. Note that by our construction, $r$ has a degree of 2 in at most $2k$ of these paths/ears. This is because for any $H_i$, $r$ is allowed to be an internal vertex in at most 2 paths in $H_i$. If $r$ appears in any of the first two paths in $H_i$, then any added path from $\{P_{i,3}, P_{i,4},...,P_{i,m_i} \}$ would truncate itself at $r$ if it ever reached $r$, hence only adding a degree of 1 to $r$. If $r$ does not appear in the first 2 paths in $H_i$, then it will have a degree of 2 in the first path it appears in out of $\{P_{i,3}, P_{i,4},...,P_{i,m_i} \}$, but any future paths will again truncate themselves at $r$ if they reach the vertex, only contributing an additional degree of 1. Either way, $r$ can only be an internal vertex in at most 2 paths in $H_i$, and hence can only have a degree of 2 in at most $2k$ paths in total.

    If $l \leq 2k$, at most $4k$ edges were used at $r$, hence at least $2k$ edges remain unused. Otherwise, when $l > 2k$, we can upper bound the degree of $r$ in $H_1 \cup H_2 \cup ...H_k$ with $(l-2k)+4k = l+2k$. Note that $l \leq d(r)/2$, where $d(r)$ denotes the number of edges incident to $r$. This is because $H_1 \cup H_2 \cup...H_k$ is a subset of $Qk$ edge disjoint paths between $v$ and $u$, and if any of those initial paths passed through $r$, it would use up two edges incident to $r$. This gives us that at least $d(r)-(d(r)/2+2k) \geq ((Q-2)/2-2)k$ edges that remain unused at $r$, which is at least $2k$ as long as $Q \geq 10$. 
\end{proof}

We now handle the cases where $|S| \geq 3$. Let $W$ be the neighbors of $v$ that are not in $S \cup R$, and define $G'=G-v-W$. Note that $G'$ retains the property that every $u \in V(G')-S-R$ has degree 3 and is adjacent to exactly 3 vertices, and $V(G')-S$ is still an independent set. We will first handle the case where $S \cup R-v$ is $6k$-edge-connected in $G'$. The idea is to prove that $6k$-edge-connectivity suffices to pack $k$ $(S \cup R-v)$-subgraphs that balance $S \cup R-v$ in $G'$, and adding $v$ and $W$ back will give the desired extension. To obtain the packing on $S \cup R -v$, we will need the lemma below. 

\begin{lemma}[Lau 3.7.9 \cite{lau-thesis}]\label{lem:treepacking}
    Let $G$ be a graph and $S$ be a subset of $V(G)$. Suppose $S$ is $3k$-edge-connected in $G$, every vertex in $V(G)-S$ is of degree 3, and there is no edge between vertices in $V(G)-S$. Let T be a set of edges with $|T| \leq k$. Then $G-T$ has $k$ edge-disjoint $S$-subgraphs.
\end{lemma}

We will use the following lemma above as a black box in several proofs. For the sake of completeness, we will also provide a proof to Lemma \ref{lem:treepacking}. We note that our proof is rather different from Lau's original proof of the lemma. We use the notion of a $(k,g)$-family, as introduced by West and Wu \cite{westwu}.

\begin{definition}
    An $S$-subgraph is called an $S$-connector if there exist a series of operations where at each step one replaces a $u,v$ path with the edge $(u,v)$ to obtain a connected graph with vertex set $S$.
\end{definition}

An $S$-connector is essentially a "stronger" version of an $S$-subgraph. For the sake of this proof, however, the only property about an $S$-connector we will really need is that it is an $S$-subgraph.

\begin{definition}
    Take a graph $G$, a set $S \subseteq V(G)$, and a function $g: V(G) \rightarrow \mathbb{Z}^+ \cup \{0 \}$ such that for all $v \in V(G)-S$, $g(v) \equiv d(v)\mod 2$. A $(k,g)$ family is a set of $k + \sum_{v \in V(G)}g(v)$ edge disjoint subgraphs, such that the first $k$ graphs form $S-$connectors, and the remaining $\sum_{v \in V(G)}g(v)$ are paths that can be oriented such that all paths end in $S$, and for each $v \in V(G)$ there are $g(v)$ paths starting at $v$. 
\end{definition}

West and Wu a necessary and sufficient condition for the existence of a $(k,g)$ family in a graph. Let $P = A_1, A_2...A_l$ be disjoint subsets of $V(G)$ such that $S \subseteq \bigcup_{i=1}^l A_i$, and for each $A_i \in P$, $A_i \cap S \not= \emptyset$. Let $B_p = V(G) - \bigcup_{i=1}^l A_i$. Let $T_p$ be the vertices in $S$ which are the only vertex in $S$ in their respective block in $P$. Let $\delta(A_i)$ denote the edges with exactly one endpoint in $A_i$, and $g(X)$ over some set $X \subseteq V(G)$ denote $\sum_{x \in X}g(x)$. 

\begin{theorem}[West-Wu \cite{westwu}]
    A $(k,g)$ family exists if and only if for all $P$, $f_g(P) = (\sum \delta(A_i))-2k(|P|-1)-g(B_p)-2g(T_p)$ is at least 0.
\end{theorem}

We now give our proof of Lemma \ref{lem:treepacking}. 

\begin{proof}
    Take a graph $G$ with $S \subseteq V(G)$ that is $3k$-edge-connected on $G$ and all vertices in $V(G)-S$ have degree exactly 3 and be an independent set. Let us delete $|T| \leq k$ edges, and we define $g:V(G) \rightarrow \mathbb{Z}$ to be 1 for $w \in V(G)-S$ that have odd degree, and 0 otherwise. Take a partition $P = A_1, A_2...A_l$ and $B_p$ of $V(G)$, and we will prove that $f_g(P) = (\sum \delta(A_i))-2k(|P|-1)-g(B_p)-2g(T_p)$ is nonzero. We first note that for any $w \in S$, we have $g(w) = 0$, so $2g(T_p) = 0$ and hence this expression reduces down to $(\sum \delta(A_i))-2k(|P|-1)-g(B_p)$.

    Now, suppose that $|P| \geq 2$. Suppose we add the edges in $T$ back to the graph. Then for each $A_i \in P$, there must be at least $3k$ edges leaving each $A_i$, because if not we could separate $A_i$ from some other $A_j$ in the partition with a cut of less than $3k$, but this would violate $S$ being $3k$-edge-connected. Hence, $\sum \delta(A_i)$ with the addition of the $T$ edges is at least $3k|P|$. Moreover, note that for each vertex $w \in B_p$, $w$ would have three edges going to $A_1,A_2...A_l$, hence each vertex contributes an increase of 3 to $\sum \delta(A_i)$, so $g(B_p)$ is at most a third of $\sum \delta(A_i)$. This then gives us $\sum \delta(A_i)-g(B_p) \geq (2/3)(3k|P|) = 2k|P|$ on the original graph without the deletion of the $T$ edges. Each edge deletion decreases $\sum \delta(A_i)$ by at most 2, and does not increase $g(B_p)$, so after deletion of the $T$ edges we have that $\sum \delta(A_i)-g(B_p) \geq 2k|P|-2k$ on the modified graph. We now have that $f_g(P) \geq 2k|P|-2k-2k(|P|-1) \geq 0$, as desired. 

    We now handle the case where $|P|=1$. Then clearly, $2k(|P|-1) = 0$. Then for each vertex $w$ in $B_p$, if $w$ has even degree, then it contributes nothing to $g(B_p)$, and if $w$ has odd degree, since $V(G)-S$ is an independent set, then $w$ has at least one edge going to $A_1$, contributing at least 1 to $\sum \delta(A_i)$. This then gives us $\sum \delta(A_i) \geq g(B_p)$ and hence $f_g(P)$ is nonzero, as desired. Either way, this allows us to find $k$ edge disjoint $S$-connectors and hence $k$ edge disjoint $S$-subgraphs on $G$, as desired. 
\end{proof}

We are now ready to eliminate the case where $S \cup R-v$ is $6k$-edge-connected in $G'$.
\begin{lemma}\label{lem:6kconn} 
    $(S \cup R-v)$ is at most $(6k-1)$-edge-connected in $G'$. 
\end{lemma}
\begin{proof}
    This proof imitates the proof of Lemma 3.7.3 in \cite{lau-thesis}. Assume that $(S \cup R-v)$ in $G'$ is $6k$-edge connected. Note that every vertex in $G'-S-R$ is adjacent to exactly 3 different vertices in $S \cup R$ and has degree exactly 3. Indeed, note that all the edges between $v \cup W$ and $G'$ are either from $v$ to a vertex in $S \cup R$ or from a vertex in $W$ which must only have neighbors in $S$, hence all vertices in $G'-S-R$ have the same edges and neighbors in $G'$ and $G$. Since $G'-S-R$ is an independent set, this means that we can find $2k$ edge disjoint $(S \cup R-v)-$subgraphs on $G'$ by Lemma \ref{lem:treepacking}. We then pair the subgraphs up into $k$ pairs of 2 subgraphs, and we union the subgraph with the other subgraph in its pair to obtain $k$ $(S \cup R-v)$-subgraphs. Note that every vertex in $(S \cup R-v)$ in each of these newly formed subgraphs now has degree 2. Denote these graphs as $H_1, H_2,...,H_k$.

    For the edges incident to $v$ that have a label $i$, add that edge to $H_i$. If the other endpoint $x$ of that edge is in $W$, then add both of the other two edges incident to $x$ to $H_i$ as well. Note that since we assumed that each label appears at least twice at $v$, each of $H_i$ now connects $S \cup R$ while extending its respective edges. Furthermore, the degree of all vertices in $S \cup R$ (including $v$) is at least 2 in $H_i$, ensuring that $H_1,H_2,...,H_k$ balance $S \cup R$, as desired. 
\end{proof}

Let $Z$ be the minimum $(S \cup R)-v$ edge cutset of $G'$, and note that $|Z| \leq 6k-1$. We can note that for the proof of Lemma \ref{lem:6kconn}, instead of packing in two $(S \cup R-v)$-subgraphs and combining them into one $S$-subgraph, we can just pack $k$ $(S \cup R-v)$-subgraphs $H_1, H_2...H_k$ instead but ensure that there are at least $k$ edges at any vertex in $S \cup R-v$ that do not appear in any of $H_1,H_2...H_k$. To accomplish this, we can use a theorem in \cite{devos2016packing}, which will yield a slightly better upper bound on $|Z|$. 

\begin{theorem}[DeVos, McDonald, Pivotto \cite{devos2016packing}]\label{thm:devosext}
    Let $G$ be a loopless multigraph, and let some $S \subseteq G$ where $|S| \geq 2$ be $2 \lceil (5k+3)/2 \rceil$-edge-connected with integer $k \geq 3$. Suppose there exists some $v \in S$ with degree exactly $2 \lceil (5k+3)/2 \rceil$. Then for any edge subpartition at $v$ where at least $k$ edges are unlabeled and each label appears once, there exist $k$ edge-disjoint $S$-subgraphs that extend $v$ and leave at least $k$ edges unused at each vertex in $S$. 
\end{theorem}

\begin{corollary}
    Let $G$ be an arbitrary loopless multigraph, and let $S$ be a $5k+4$-edge-connected subset of $V(G)$. Then there exist $k$ edge disjoint $S$-subgraphs that balance $S$. 
\end{corollary}

\begin{proof}
    We can assume $|S| \geq 2$ otherwise this is trivial. Delete edges from $G$ until $S$ is minimally $2 \lceil (5k+3)/2 \rceil$-edge connected. Take a minimum edge cut $Z'$ of $S$ and break the remaining graph into components $C_1$ and $C_2$. Condense $C_1$ and $C_2$ to $v_1$ and $v_2$ respectively, and examine $C_1+v_2$ and $C_2+v_1$. $v_2$ is clearly $2 \lceil (5k+3)/2 \rceil$-edge-connected to $S \cap C_1$, as there must be $2 \lceil (5k+3)/2 \rceil$-edge-disjoint paths from $C_1$ to $C_2$. Note that $|Z'|=2 \lceil (5k+3)/2 \rceil$ since we assumed minimal connectivity. Using DeVos's extension theorem, we can pack $k$ edge disjoint $(S \cap C_1) \cup v_2$-subgraphs in $C_1+v_2$ such that each vertex in $(S \cap C_1) \cup v_2$ has at least $k$ incident edges unused, and hence is balanced. We can then label the edges at $v_1$ accordingly, noting that all $k$ labels must appear once on $v_1$ and $v_1$ must have $k$ unlabeled edges. This allows us to apply the same extension theorem to $C_2+v_1$ to get $k$ edge disjoint $(S \cap C_2) \cup v_1$-subgraphs that extend $v_1$ and balance $(S \cap C_2) \cup v_1$. By Lemma \ref{lem:strict-extn}, we can assume that the extension at $v_1$ is a strict extension while retaining the properties listed above. Combine the packings by taking the edge-union of the $i$th subgraph on $C_2+v_1$ and the $i$th subgraph on $C_1+v_2$. We then obtain $k$ edge disjoint $S$-subgraphs that balance $S$, as desired.
\end{proof}

\begin{lemma}\label{lem:5.5conn} 
    If $k \geq 8$, $(S \cup R-v)$ is at most $(\lfloor 5.5k \rfloor-1)$-edge-connected in $G'$. 
\end{lemma}

\begin{proof}
    Assume that $S \cup R-v$ in $G'$ is $\lfloor 5.5k \rfloor$-edge connected. Since $k \geq 8$, we know that $S \cup R-v$ in $G'$ is at least $5k+4$ edge-connected. Then, by the previous lemma, we can pack $k$ edge disjoint $(S \cup R)-v$ subgraphs that balance $(S \cup R)-v$ on $G'$. Then extending $v$ in $G$ is easy, giving us $k$ edge disjoint $S \cup R$ subgraphs which balance $S\cup R$ and extend $v$. 
\end{proof}

It follows that when $k \geq 8$, we have that $|Z| \leq 5.5k-1$. We will still use the bound $|Z| \leq 6k-1$ for the rest of this proof, but using the bound $|Z| \leq 5.5k-1$ will yield a bound of $31k$ rather than $32k$ on the Steiner Forest packing problem with the analysis otherwise the same.

\subsection{Extension Theorem Proof Part 3: Constructing $S$-subgraphs with $|Z| < 6k$}

We will now handle the case where $S \cup R-v$ in $G'$ is not highly connected, or when $|Z| \leq 6k-1$, which will resolve the proof of the Extension Theorem. To do this we delete $Z$ to split $S \cup R-v$ into two highly edge-connected components, and we pack $k$ subgraphs on both components and merge them together for the final packing. We will hence need a few lemmas establishing the structure of $G'$ before we pack the $S$-subgraphs. We will first start by proving the deletion of $Z$ will only split $G'$ into 2 connected components. 

To assist us in proving this, we will first construct a collection of paths for each $x \in S \cup R-v$ with the properties below. 

\begin{lemma}\label{lem:P(x)}
    Given a vertex in $x \in S \cup R-v$, there exists a collection $P(x)$ of $(Q-4)k$ edge disjoint paths in $G$ from $v$ to $x$ that remain connected after deletion of $v$ and $W$.  
\end{lemma}

\begin{proof}
    Note there must exist a collection $(Q-4)k$ edge disjoint paths from $v$ to $x$ on $G$ by Lemma \ref{lem:Q-4-conn}, and let $P(x)$ be such a collection of $(Q-4)k$ edge disjoint paths such that the sum of lengths of the paths is minimized. Now, assume, for contradiction, that such a path becomes disconnected after deletion of $v$ and $W$. Then such a path has a vertex $w \in W$ as an internal vertex which was not adjacent to $v$ on the path. Then $vw$ could not have been used in $P(x)$ as all vertices in $W$ have a degree of 3 in $G$, so we can shorten this path by cutting off the first part of the path that starts at $v$ and traverses to $G'$ and then to $w$, by simply replacing it with $vw$. Hence $P(x)$ exists for any $x \in S \cup R-v$, as desired. 
\end{proof}

For each vertex $x \in S \cup R-v$, assign $P(x)$ to be a collection of paths as described in the previous lemma. Define $P'(x)$ to be the truncated collection of paths in $P(x)$ after the deletion of $v$ and $W$, that is, $P'(x) = \{P-v-W | P \in P(x)\}$. Lau refers to these paths as \emph{diverging paths} in his thesis. 

\begin{lemma} 
    $G'-Z$ has 2 connected components. 
\end{lemma}

\begin{proof}
    This proof imitates the proof Lemma 3.7.4 in \cite{lau-thesis}. Note that it is sufficient to prove that $G'$ has at most two components. This is because if $G'$ has exactly two components, $|Z|$ will be empty and this will trivially be true. Otherwise, if $G'$ is connected, $G'-Z$ will have exactly two components by minimality of $Z$. Take any component of $G'$, and note that for each component, there must be a vertex that was incident to $v$ or $W$ on $G$. Since vertices in $W$ are only incident to vertices in $S$ by Lemma \ref{lem:indset}, and all neighbors of $v$ in $G'$ are in $S \cup R$, each connected component must contain at least one vertex in $S \cup R$. 

    Suppose that there are at least three components in $G'$, and select $u_1,u_2,u_3$ each from a different component such that all three are vertices in $S \cup R$. Then take the collection $P(u_i)$ of $(Q-4)k$ edge disjoint paths from $u_i$ to $v$ as defined in Lemma \ref{lem:P(x)}. Since $v$ has $Qk$ incident edges in $G$, there must be some common vertex $w \in N_G(v)$ which $u_1,u_2,u_3$ all end in $vw$, as long as $Q \geq 13$. If $w \in S \cup R$, then $u_1,u_2,u_3$ are still all connected in $G'$, which is a contradiction. Otherwise, we have that $w \not \in S \cup R$ and hence $w \in W$. Those three paths must've traversed through a neighbor of $w$ that is not $v$, but since $w \not \in S \cup R$, then $w$ must have three neighbors, including $v$. Hence, two paths will need to traverse through the same neighbor of $w$ to reach $v$, and hence still remain connected in $G'$, a contradiction. 
\end{proof}

We denote the components as $C_1$ and $C_2$. For each $w_i \in W$, note that it has exactly two neighbors $x_i$ and $y_i$ in $G'$, which we call a couple. Let $X_1$ be the collection of couples where both vertices are in $C_1$, $X_2$ be the collection of couples where both vertices in $C_2$, and let $X_C$ be the collection of couples with one vertex in $X_1$ and another vertex in $X_2$. We refer to the couples in $X_C$ as \emph{crossing couples}. Let $B_i$ be $C_i \cap N(v) \cap (S \cup R)$, and let $c(B_i)$ denote the number of edges between $v$ and $B_i$. Crossing couples are important in our proof, as they can be used to connect vertices in $C_1$ and $C_2$ when constructing our $S$-subgraphs. We therefore wish to prove that we can find a large number of crossing couples in $G$, which we do so with the next two lemmas.

\begin{lemma}
    $C_1$ and $C_2$ both contain vertices in $S$.
\end{lemma}

\begin{proof}
    Assume not, and without loss of generality let $C_2$ only contain vertices not in $S$. Since $V(G)-S$ is an independent set, then $C_2$ must consist of a single vertex from $r \in R$. Then, $r$ cannot be adjacent to any vertex in $W$, which means $r$ can only have incident edges in $Z$ and edges incident to $v$ in $G$. Since $r$ has at least $(Q-4)k$ incident edges that are not self-loops in $G$, we have that there are at least $(Q-4)k-|Z|$ edges between $v$ and $r$, leaving at most $4k+|Z|$ edges incident to $v$ that are not incident to $r$. Now, take some $s_1 \in C_1$, and note that we can find $(Q-4)k$ edge disjoint paths between $s_1$ and $r$ in $G$. All such paths must either traverse through $Z$ or $v$ to reach $r$, and if such a path traverses through $v$, it must traverse an edge incident to $v$ but not incident to $r$. Hence, at most $4k+2|Z| \leq 16k$ paths must exist. But as long as $Q > 20$, then there were more than $16k$ paths between $s_1$ and $r$, a contradiction. Hence $C_1$ and $C_2$ must both contain vertices in $S$. 
\end{proof}

\begin{lemma}\label{lem:crossingboundS}
 $|X_C| \geq Qk-2|Z|$. 
\end{lemma}

\begin{proof}
    Take some $u_1 \in S \cap C_1$ and note that there are $Qk$ edge disjoint paths from $v$ to $u_1$ in $G$. Note that there are $|X_2|+c(B_2)$ of these paths that start from $v$ and go to $C_2$ before reaching $u_1$. Hence, these paths must either traverse through $Z$ or traverse through $W$ to reach $C_1$. However, if such a path reaches $C_1$ from $C_2$ by traversing through some $w \in W$, then this blocks off other paths from using the edge $vw$ as $w$ has degree 3 in $G$. Moreover, this path would also block off $vw'$, where $w'$ was the first vertex traversed after $v$ in the path. Since $w'$ is traversed before or during the first time this path enters $C_2$, and $w$ is traversed after the path leaves $C_2$, then $w' \not= w$. This means that the remaining $Qk-1$ edge-disjoint $v,u_1$-paths can only utilize $Qk-2$ incident edges at $v$, a contradiction. Hence, these $|X_2|+c(B_2)$ paths that initially traverse to $C_2$ must all pass through $Z$ at some point to reach $u_i$, so $|Z| \geq |X_2|+c(B_2)$. Then, by a similar argument, $|Z| \geq |X_1|+c(B_1)$ as well. Note that $Qk = |X_C|+|X_1|+|X_2|+c(B_1)+c(B_2)$, giving us $Qk \leq |X_C|+2|Z| $ and hence $|X_C| \geq Qk-2|Z|$, as desired.
\end{proof}




Now that we have established that we have a large number of crossing couples to help merge the subgraphs in both components together, we now need to prove that in each component $S \cup R$ is highly edge-connected, so we can construct the packing of $k$ subgraphs in each component. To do this, we introduce the definition of a \emph{common path}.

\begin{definition}
    We say $v_1$ and $v_2$ have $\lambda$ \emph{common paths} if there exist $\lambda$ edge disjoint paths starting from $v_1$, $\lambda$ edge disjoint paths starting from $v_2$, and a bijection between the paths such that each pair of paths ends at the same vertex. We allow the empty path (from a vertex to itself) to count as a path, as long as the other vertex has a corresponding path to that vertex. 
\end{definition}

The following lemma shows us how to get edge connectivity from a set of common paths. This was proven in Lau's thesis and used as a black box in our proof, however, we will prove it again for the sake of completeness.

\begin{lemma}[Lau 3.7.5 \cite{lau-thesis}]\label{lem:commonconn}
    For any arbitrary graph $G$, if $v_1$ and $v_2$ have $2\lambda+1$ common paths in $G$, there exist $\lambda+1$ edge disjoint paths from $v_1$ to $v_2$ in $G$. 
\end{lemma}

\begin{proof}
    Suppose for contradiction, that there is a edge cut of size at most $\lambda$ from $v_1$ to $v_2$. Note that deleting any edge removes at most 2 of the common paths between $v_1$ and $v_2$. Then, deleting all the edges part of the edge cut means that there is at least $2\lambda+1-2\lambda = 1$ common path left, and hence $v_1$ and $v_2$ are still connected, a contradiction. 
\end{proof}

We will demonstrate that $S \cup R$ is highly edge-connected in both $C_1$ and $C_2$ by finding a set of common paths. We make use of $P$ and $P'$ as constructed in Lemma \ref{lem:P(x)}. Let $S_1$ and $S_2$ denote $S \cap C_1$ and $S \cap C_2$ respectively. Similarly, let $R_1$ and $R_2$ denote $R \cap C_1$ and $R \cap C_2$. 

\begin{lemma}\label{lem:commonpaths}
    Given $i \in \{1,2\}$, For any $a,b \in S_i \cup R_i$ where $a \not=b $, there exist $(Q-20)k$ common paths between $a$ and $b$ in $G'-Z$. 
\end{lemma}

\begin{proof}

     This proof imitates the proof of Lemma 3.7.7 in \cite{lau-thesis}. We can assume, without loss of generality, that $a$ and $b$ are in $C_1$. We aim to find common paths by using the neighbors of $W$ in $C_1$, and the endpoints of $Z$. Let $k_a$ denote the number of crossing couples which $P'(a)$ does not have a corresponding path to - that is, the crossing couple adjacent to $w \in W$ does not have a path in $P(a)$ that starts with $vw$. Let $z_a$ denote number of the crossing couples which $P'(a)$ does have a corresponding path to, but the path traverses through $Z$. Note that if a path in $P(a)$ first traverses to some $w \in W$ adjacent to a crossing couple after $v$, but then goes to $C_2$ right after, it must pass through $Z$ - so we can assume that if a path in $P'(a)$ corresponds to a crossing couple and doesn't pass through $Z$, it must start at a vertex belonging to a crossing couple in $C_1$. Then $a$ has $|X_C|-k_a-z_a$ paths to a crossing couple in $C_1$. Hence, $a$ and $b$ share at least $|X_C|-k_a-z_a-k_b-z_b$ common paths in $C_1$ which all end at a vertex in $C_1$ belonging to a crossing couple. 

    We now count the common paths $a$ and $b$ share through the endpoints in $Z$. Consider the paths in $P(a)$ that use an edge incident to $v$ corresponding to $X_2$ or an edge incident to $v$ that has its other endpoint in $C_2$. All such paths must pass through $Z$, so truncate the paths the first time one would finish traversing $Z$ if they started from $a$. There are $|X_2|+c(B_2)-m_a$ such paths, where $c(B_2)$ is the number of incident edges in $v$ that has its other endpoint in $C_2$, and $m_a$ is for the number of such edges where $P'(a)$ has no corresponding path to. We add back the $z_a$ paths we ignored in the previous paragraph, and truncate them in the same manner. Note that the expression $(|X_2|+c(B_2)-m_a+z_a) + (|X_2|+c(B_2)-m_b+z_b)$ counts the edges shared by $P'(a)$ and $P'(b)$ in $Z$ exactly twice, and other edges in $Z$ less than twice, so $2|X_2|+2c(B_2)-m_a-m_b+z_a+z_b-|Z|$ is a lower bound on the common paths between $a$ and $b$ that end by traversing the same edge in $Z$. For these common paths to be entirely contained in $C_1$, we truncate them right before they traverse $Z$.

    We have $(|X_C|-k_a-z_a-k_b-z_b)+(2|X_2|+2c(B_2)-m_a-m_b+z_a+z_b-|Z|)$ as a lower bound on the common paths between $a$ and $b$. Note that by construction of $P'(x)$, $k_x+m_x \leq 4k$ for any $x \in S_1 \cup R_1$, so this simplifies down to $|X_C|+2|X_2|+2c(B_2)-8k-|Z|$. 

    We now claim that $|X_C|+c(B_2)+|X_2| \geq Qk-|Z|$. Indeed, for any $s_1 \in S_1$ and $s_2 \in S_2$, there are $Qk$ edge disjoint paths between the two vertices. $s_2$ can use at most $|Z|$ paths to reach $s_1$ without $v$ or $W$ by virtue of these paths being edge-disjoint, but otherwise needs to traverse through $v$ or a crossing couple. More specifically, any path from $s_1$ to $s_2$ in $G$ that does not pass through $Z$ will either have to go through an edge between $X_C$ and $C_2$ (which there exist $|X_C|$ of), or an edge between $v$ and $X_2$ (which there exist $|X_2|$ of), or an edge between $v$ and $C_2$ (which there exist $c(B_2)$ of). As at least $Qk-|Z|$ edge-disjoint paths need to go through one edge from a collection of $|X_C|+|X_2|+c(B_2)$ edges, this would imply that $|X_C|+|X_2|+c(B_2) \geq Qk-|Z|$. 
    We now have that $|X_C|+2|X_2|+2c(B_2)-8k-|Z| \geq Qk-8k-2|Z|$, and since $Z < 6k$, this gives us at least $(Q-20)k$ common paths. 
\end{proof}

We are finally ready to construct the $S$-subgraphs that extend $v$ for the case where $|Z| \leq 6k-1$. 

\begin{lemma}
    $S_i \cup R_i$ is $(Q/2-10)k$-edge-connected in $G'-Z$. 
\end{lemma}
\begin{proof}
    Follows as a corollary from the previous lemma and Lemma \ref{lem:commonconn}.
\end{proof}

\begin{lemma}\label{lem:case1}
    $G$ has $k$ edge disjoint $S$ subgraphs that balance $S \cup R$ and extend $P_k(v)$. 
\end{lemma}

\begin{proof}
    This proof imitates the proof of Lemma 3.7.10 in \cite{lau-thesis}. We first handle the case where $|S_i \cup R_i| \geq 2$ for $i \in \{1,2\}$. Note that if a vertex $u \in V(G')-S-R$ is incident to an edge in $Z$, by minimality of $Z$, the other two edges incident to $u$ must be contained in the same component $C_i$ that $u$ is in, which we call $ux_1$ and $ux_2$. To apply Lemma \ref{lem:treepacking} to $C_i$, we first replace $ux_1$ and $ux_2$ with $x_1x_2$, which we note does not alter the connectivity nor the degree of $S_i \cup R_i$. We now have that every vertex in $C_i-S_i-R_i$ has degree 3 and exactly 3 neighbors. 

    Choose $\min(k,|Z|)$ edges from $Z$, and call them reserve edges. Note that each reserve edge may have one endpoint in $V(G')-S-R$. If such an endpoint $u$ exists, then we delete the edge $x_1x_2$ from $C_i$ we "split off" from $u$, or in other words, $x_1x_2$ is the edge we replaced $ux_1$ and $ux_2$ with and we delete $x_1x_2$. Note that we delete at most $k$ edges from each $C_i$ in this manner. Call the altered component $C_i'$. Since all the deleted edges had both endpoints in $S_i \cup R_i$, and $Q/2-10 \geq 6$ as long as $Q \geq 32$, we can apply Lemma \ref{lem:treepacking} and find $2k$ edge disjoint $S_1 \cup R_1$-subgraphs $A_1',A_2',...,A_{2k}'$ on $C_1'$, and $2k$ edge disjoint $S_2 \cup R_2$ subgraphs $B_1',B_2',...,B_{2k}'$ on $C_2'$. Set $A_1,A_2,...,A_k$ as $A_i = A_i' \cup A'_{i+k}$ and $B_i = B_i' \cup B'_{i+k}$. Clearly, $A_1,A_2,...,A_k$ are $S_1 \cup R_1$ subgraphs that balance $S_1 \cup R_1$, as the degree of any vertex in $S_1 \cup R_1$ in any $A_i$ is at least 2. Similarly, subgraphs $B_1,B_2,...,B_k$ are $S_2 \cup R_2$ subgraphs that balance $S_2 \cup R_2$ by construction. If $x_1x_2 \in A_j$ and was an edge we split off from some vertex $u$, delete $x_1x_2$ and add $x_1u$ and $x_2u$ to $A_j$. We repeat this for $C_2$ as well. 

    If $k \leq |Z|$, then we have $k$ reserve edges. Label them $e_1,e_2,...,e_k$, and set $H_j = A_j \cup B_j + e_j$. If both endpoints of $e_j$ are in $S \cup R-v$, then $H_j$ is connected and hence connects $S \cup R$. Now if there is some endpoint $u$ of $e_j$ that is not in $S \cup R$, then its two other edges were split off and deleted previously. We add $x_1u,x_2u$ to $H_j$, where $x_1u$ and $x_2u$ were the edges split off at $u$. Note that $x_1,x_2 \in S \cup R$, so applying this procedure to the endpoints as needed ensures that each $H_j$ connects $S \cup R$. We note that by construction, there is no overlap in edges between $H_i$ and $H_j$ where $i \not=j$, and $H_j$ is connected, as desired. We then add back the $v$ and the vertices in $W$, and if $vw$ had a label of $j$ with $w \in W$, we add all edges incident to $w$ to $H_j$. For the remaining labeled edges of $v$, we add that edge to the respective subgraph. Then $H_1,H_2....H_k$ extend $v$, connect $S \cup R$, and balance $S \cup R$ by virtue of each vertex in $S \cup R$ having a degree of 2 in any $H_i$.

    Now suppose $k > |Z|$. In this case, our reserve edges are not enough for connectivity, so we need to utilize the crossing couples. We first define $H_j = A_j \cup B_j$, and for any edge $vw$ with a label of $j$, we add all 3 edges incident to $w$ to $H_j$. If $w$ is adjacent to $x,y$ which form a crossing couple (we say in this case that $j$ \emph{corresponds} to a crossing couple), then $A_j$ and $B_j$ are connected in $H_j-v$, as they can reach each other by crossing $xw$ and $yw$. Hence we get the desired connectivity and extension properties of $H_j$. Now our only issue is the labels which do not correspond to crossing couples. By Lemma \ref{lem:crossingboundS} we have that $|X_C| \geq Qk-2|Z|$, and since each label appears twice at $v$, there can be at most $|Z|$ labels missing from the crossing couples. For each missing label, we utilize the reserve edges to connect $A_j$ and $B_j$, and repeat the previous procedure we did when $k \leq |Z|$ to obtain $H_j$. For any remaining edges of $v$ with label $i$, we add that edge to $H_i$, and if the other endpoint is in $W$ we add its other two incident edges to $H_i$ as well, which completes our construction. 

    Now we handle the case where there exists a component that has exactly one vertex in $S \cup R$, without loss of generality let that component be $C_1$. Let us first assume that $C_2$ has at least 2 vertices in $S \cup R$. Note that $C_1$ consists of only one vertex in $S$, which we call $s$, because if that vertex was incident to some vertex $u \in V(G')-S-R$ which was also part of $C_1$, then two edges incident to $u$ would need to be in $Z$, but we could get a smaller cut by moving $u$ to $C_2$, contradicting the minimality of $Z$. Construct $B_1,B_2,...,B_k$ with the same method outlined previously, and we use exactly $k$ reserve edges and crossing couples to connect $s$ to $B_1,B_2,...,B_k$. For the remaining crossing couples, instead of adding both edges incident to $C_1$ and $C_2$ to the corresponding graph in $B_1,B_2,...,B_k$, only add the edge that is incident to $C_2$. Finish by adding back the remaining vertices in $v$ and $W$ as outlined previously. There are at least $Qk-2(6k)=(Q-12)k$ crossing couples, $k$ of which we might have used to connect $s$ to the rest of the graph, leaving at least $k$ edges from crossing couples unused at $s$, hence balancing $s$, as desired. 

    Finally, say both $C_1$ and $C_2$ only have one vertex in $S \cup R$. Then $C_1$ and $C_2$ each consist of a single vertex, $s_1$ and $s_2$, as all vertices outside of $S \cup R$ in $G'$ have 3 neighbors, but this can't happen when $|S|$ has size 2. Similar to what was done previously, we use at most $k$ crossing couples and edges from $Z$ to connect $s_1$ and $s_2$. Then, for the remaining crossing couples, split them into two groups of equal size $Y_1$ and $Y_2$ (or allow a difference of 1 if there is an odd number). For any labeled crossing couple in $Y_1$, we add the edge incident to $v$ and the edge incident to $s_1$ to the respective subgraph. Similarly, for any labeled crossing couple in $Y_2$, we add the edge incident to $v$ and the edge incident to $s_2$ to the respective subgraph. This means that each vertex $s_1,s_2$ has at least $\lfloor (Q-13)k/2 \rfloor > 2k$ incident unlabeled edges and is guaranteed balancing. We add back the remaining edges incident to $v$ accordingly and we are done. 
\end{proof}

This concludes the proof of the Extension Theorem. 

\section{Steiner Forest Packing}
\begin{lemma}\label{lem:basecase}
    Let $G$ be a multigraph that admits self loops. Suppose there is $S,R \subseteq V(G)$ where $S \cap R = \emptyset$, $|S| \geq 2$, $S$ is $32k$-edge-connected and every vertex in $R$ is incident to at least $30k$ edges. Then there exist $k$ edge disjoint $S$-subgraphs that balance $S \cup R$. 
\end{lemma}

\begin{proof}
    Delete edges from $G$ until $S$ is minimally $32k$-edge-connected. If a deleted edge was incident to some $r \in R$, add a self loop to $r$ to preserve the number of edges $r$ is incident to. Then we can find a $S$-cut of exactly $32k$ edges, and denote the components $C_1$ and $C_2$. Contract $C_1$ and $C_2$ into $v_1$ and $v_2$ respectively. Let $G_1$ denote the graph formed by $v_2+C_1$ and $G_2$ denote the graph formed by $v_1+C_2$. 

    Note that by construction, $v_1$ has degree $32k$, is incident to no self loops, and is $32k$-edge-connected to $S \cap C_2$ on $G_2$. Hence, we can find $k$ edge disjoint $v_1 \cup (C_2 \cap S)$-subgraphs that balance $v_1 \cup (C_2 \cap (S \cup R))$ on $G_2$ by the Extension Theorem. These subgraphs induce a balanced subpartition at $v_2$ on $G_1$, and furthermore, we note that each label must appear at least once on $v_2$, as $v_1$ was connected to $(C_2 \cap S)$ in each subgraph and hence has an incident edge in all $k$ subgraphs, which would be copied over to the labels at $v_1$. Since $v_2$ is also $32k$-edge-connected to $C_1$, has degree $32k$ and is incident to no self loops. Then we can find $k$ edge disjoint $v_2 \cup (C_1 \cap S)$-subgraphs that extend $v_2$ and balance $v_2 \cup (C_1 \cap (S \cup R))$ on $G_1$. By Lemma \ref{lem:strict-extn}, we can assume without loss of generality that this extension is a strict extension.
    
    Then we can take the edge-union of the $i$th subgraph on $G_1$ and the $i$th subgraph on $G_2$ to give us an $S$-subgraph on $G$, hence giving us $k$ edge disjoint $S$-subgraphs. We note that $S \cup R$ is balanced after the edge union, as none of the vertices had their edge labels changed during the edge-union by virtue of using a strict extension. Note that we can assume, without loss of generality, that the added self loops were not used in the packing, as self loops do not affect the connectivity and deleting labels at incident edges of a balanced vertex will still keep a vertex balanced. Hence $R$ will still be balanced after we replace the self loops with their original edges. 
\end{proof}

\begin{theorem}[Steiner Forest Packing]
    Given a loopless multigraph $G$, $S_1,S_2,...,S_t, R \subseteq V(G)$ where $(S_1 \cup S_2...S_t) \cap R = \emptyset$, such that each $S_i$ is $32k$-edge-connected, $|S_i| \geq 2$, and all vertices in $R$ are incident to at least $30k$ edges. Then there exist $k$ edge-disjoint $(S_1,S_2,...,S_t)$-subgraphs that balance $R$. 
\end{theorem}

\begin{proof}
    Suppose, for contradiction, that this is false. Then take a counterexample with $|V(G)|+|E(G)|$ minimized. We can assume that $S_1 \cup S_2...S_t$ is not $32k$-edge-connected, otherwise this problem reduces down to the previous lemma.

    Let $C$ be a subset of $V(G)$ such that the number of edges between $C$ and $V(G)-C$ is at most $32k$, for all $1 \leq i \leq t$, either $S_i \subseteq C$ or $S_i \subseteq V(G)-C$, and there exists $i,j$ such that $S_i \subseteq C$ and $S_j \subseteq V(G)-C$. We know such a $C$ must exist, since there is an $S_1 \cup S_2...S_t$-cut of less than $32k$, which keeps all $S_i$ intact since each are $32k$-edge connected. Pick a $C$ with the minimum possible number of vertices, and we will use $C_2$ to denote $V(G)-C$. Contract $C_2$ into a vertex $v_2$ and $C$ into a vertex $v$ to obtain $G_2 = C_2+v$ and $G_1 = C+v_2$. 

    If $v$ has degree at least $30k$, add it to $R$ so we can ensure its balancing. Note that since $|C| \geq 2$, $G_2$ is strictly smaller than $G$, so we can inductively find $k$ edge disjoint subgraphs, each of them which connect each $S_i \subseteq C_2$ and balance each $R \cap C_2$ on $G_2$. 

    Let $S^*$ denote $(S_1 \cup S_2...S_t) \cap C$. We first claim $S^*$ is $32k$-edge-connected in $G_1$. Indeed, if there is a cut of $< 32k$ that splits $C$ into $C_a$ and $C_b$, it must keep all $S_i \subseteq C$ intact. Without loss of generality suppose $v_2 \in C_b$, then we can replace $C$ with $C_a$ to get a smaller set that satisfies the desired properties of $C$, contradicting the minimality of $C$.

    We next claim that there are $d(v_2)$ edge disjoint paths from $v_2$ to any $s \in S^*$. Suppose, for contradiction, that there was a cut of $< d(v_2)$ with $s \in C_a$ and $v_2 \in C_b$. Then clearly, $S^* \subseteq C_a$, and $v_2$ cannot be a single vertex since its degree is $d(v_2)$, so $|C_b| \geq 2$. Then this again contradicts the minimality of $C$. Based on the forest packing on $G_2$, copy the corresponding labels over to $v_2$ on $G_1$. Then, pick any vertex $s' \in S^*$, and add $32k-d(v_2)$ fake edges from $v_2$ to $s'$. 

    We claim that $v_2$ is now $32k$-edge-connected to any $s \in S^*$. Indeed, suppose there is a cut of $< 32k$ separating $v_2$ and $s$. Since the cut is less than $32k$, $s'$ and $s$ will still be connected, so it means that all the fake edges must have been deleted. Then this implies that the deletion of less than $d(v_2)$ edges disconnects $v_2$ and $s$, contradicting what we found previously. Assign the fake edges to have no label. Note that if $d(v_2) \geq 30k$, then $v_2$ is balanced as $v$ is balanced with respect to the packing on $G_2$. Otherwise, $v_2$ is balanced with the addition of the fake edges, since if $d(v_2) < 30k$, we added at least $2k$ fake edges. 

    Apply the extension theorem to get $k$-edge-disjoint $S^* \cup v_2$ subgraphs $H_1,H_2,...,H_k$ that balances $R \cap C$. We now delete the fake edges we added from the graph entirely. Furthermore, to make this extension a strict edge extension, if an edge incident to $v_2$ was placed in $H_i$ in the packing, yet it was not initially assigned a label of $i$, we remove the edge from $H_i$ but keep it in the graph. Note that since each subgraph $H_i$ is still connected after the deletion of $v_2$, all of $H_i-v_2$ is still connected after the deletion or removal of those edges. Furthermore, since none of the fake edges were incident to a vertex in $R$, and removing edge labels does change a balanced vertex to a unbalanced vertex, $C \cap R$ is still balanced. Let $H_1',H_2',...,H_k'$ be the $k$ edge disjoint subgraphs we found on $G_2$. We note that $H_1 \uplus H_1', H_2 \uplus H_2'... H_k \uplus H_k'$ balances $R$ and are edge disjoint by construction. We claim that $H_i \uplus H_i'$ is a $(S_1,S_2,...,S_t)$-subgraph. 

    Note that if any $S_j \subseteq C$, then $S_j$ is connected in $H_i \uplus H_i'$ as $H_i-v_2$ is connected. For any $S_j \subseteq C_2$, given $a, b \in S_j$, if there exists a path from $a$ to $b$ that does not traverse through $v$ in $H_i'$, then $a$ and $b$ are connected on $H_i'$ as well. Otherwise, if all paths from $a$ to $b$ traverse through $v$ in $H_i'$, take a path $P_1$ from $a$ to $v$ on $H_i'$, and a path $P_2$ from $v$ to $b$ on $H_i'$. Let $e$ be the last edge of $P_1$ and $e'$ be the first edge of $P_2$. Then by construction, $e$ and $e'$ appear on $H_i$ (even after deletion of the fake edges and removal of edges to get a strict extension). Setting $e = uv_2$ and $e' = u'v_2$, since $H_i-v_2$ is connected, we can find a path from $e$ to $e'$ on $H_i-v_2$ and hence $H_i \uplus H_i'$. Then, we can traverse from $a$ to $b$ on $H_i \uplus H_i'$ by first following the $P_1$ from $a$ to $u$, traversing from $u$ to $u'$, and then following $P_2$ to reach $b$. Hence, we have found $k$ edge-disjoint $(S_1,S_2,...,S_t)$-subgraphs that balance $R$, as desired. 
\end{proof}

\section*{Acknowledgments}
This work was done as my undergraduate thesis. I thank my advisor, Chandra Chekuri, for his support, discussions, and feedback on this research and this manuscript. I thank Lap Chi Lau for the suggestions he provided on fixing his proof. I also thank the anonymous reviewers for giving suggestions to improve this manuscript, as well as a simplification of the proof that brought the initial bound of $36k$ down to $32k$.

\bibliography{references}
\end{document}